\documentclass[a4paper, UKenglish, cleveref, autoref, thm-restate]{lipics-v2021}

\hideLIPIcs
\nolinenumbers

\usepackage{etoolbox}
\usepackage{graphicx}
\usepackage{pdfpages} 
\usepackage{amsfonts}
\usepackage{amsmath}
\usepackage{amssymb}
\usepackage{amsthm}
\usepackage{xfrac}
\usepackage{nicefrac}
\usepackage{bm}
\usepackage{bbm}
\usepackage{mathtools}
\usepackage{mdframed}

\usepackage{zref-savepos}

\usepackage{graphicx}
\graphicspath{{graphics/}}

\usepackage{xspace}
\usepackage[usenames,dvipsnames]{xcolor}
\usepackage{mleftright}
\usepackage{forloop}
\definecolor{darkgreen}{rgb}{0,0.5,0}
\usepackage{hyperref}
\hypersetup{
    unicode=false,            
    colorlinks=true,          
    linkcolor=blue!70!black,  
    citecolor=darkgreen,      
    filecolor=magenta,        
    urlcolor=blue!70!black    
}
\usepackage{mleftright}

\usepackage{algorithm}
\usepackage[noend]{algpseudocode}

\newfloat{protocol}{htbp}{loa}
\floatname{protocol}{Protocol}
\makeatletter\newcommand\l@protocol{\@dottedtocline{1}{1.5em}{2.3em}}\makeatother

\usepackage[immediate]{silence}
\usepackage{sectsty}
\DeactivateWarningFilters[temp]

\allsectionsfont{\boldmath}

\usepackage{anyfontsize}

\crefname{theorem}{Theorem}{Theorems}
\Crefname{lemma}{Lemma}{Lemmas}
\Crefname{claim}{Claim}{Claims}
\Crefname{fact}{Fact}{Facts}
\Crefname{remark}{Remark}{Remarks}
\Crefname{observation}{Observation}{Observations}
\Crefname{line}{Line}{Lines}
\crefalias{AlgoLine}{line}

\makeatletter
\newcounter{algorithmicH}
\let\oldalgorithmic\algorithmic
\renewcommand{\algorithmic}{%
  \stepcounter{algorithmicH}
  \oldalgorithmic}
\renewcommand{\theHALG@line}{ALG@line.\thealgorithmicH.\arabic{ALG@line}}
\makeatother

\newcommand{\e}{\varepsilon}
\newcommand{\eps}{\e}

\newcommand{\prob}[1]{\mathbb P \left[ #1 \right]}

\newcommand{\RSet}{{\cR}}
\DeclareMathOperator{\Span}{span}

\newcommand{\aall}{{a_{\mathrm{all}}}}

\newcommand{\defcal}[1]{\expandafter\newcommand\csname c#1\endcsname{{\mathcal{#1}}}}
\newcommand{\defbb}[1]{\expandafter\newcommand\csname b#1\endcsname{{\mathbb{#1}}}}
\newcounter{calBbCounter}
\forLoop{1}{26}{calBbCounter}{
    \edef\letter{\Alph{calBbCounter}}
    \expandafter\defcal\letter
		\expandafter\defbb\letter
}

\newcommand{\nnR}{{\bR_{\geq 0}}}
\DeclareMathOperator{\rank}{rank}

\DeclareMathOperator{\val}{Val}

\newcommand{\OPT}{\ensuremath{\mathrm{OPT}}\xspace}
\newcommand{\OPTbig}{\ensuremath{\mathrm{OPT_{big}}}\xspace}

\algnewcommand\myand{\textbf{and} }
\algnewcommand\myor{\textbf{or} }

\newcommand{\characteristic}{{\mathbf{1}}}

\newcommand{\ALG}{{\text{ALG}}}

\newcommand{\DSCG}{{\texttt{DSCG}}}
\newcommand{\SMM}{{\texttt{SMMatroid}}}
\newcommand{\MSMM}{{\texttt{MSMMatroid}}}

\newbool{anon}

\makeatletter
\@ifclasswith{lipics-v2021}{anonymous}{\booltrue{anon}}{\boolfalse{anon}}
\makeatother



\authorrunning{M. Feldman, P. Liu, A. Norouzi-Fard, O. Svensson, and R. Zenklusen} 
\Copyright{Moran Feldman, Paul Liu, Ashkan Norouzi-Fard, Ola Svensson, and Rico Zenklusen} 

\title{Streaming Submodular Maximization under Matroid Constraints}

\ifbool{anon}{
}{
\author{Moran Feldman}{University of Haifa \and \url{cs.haifa.ac.il/~moranfe}}{moranfe@cs.haifa.ac.il}{https://orcid.org/0000-0002-1535-2979}{\flag{isf_logo}Research supported in part by the Israel Science Foundation (ISF) grants no. 1357/16 and 459/20.}

\author{Paul Liu}{Stanford University \and \url{cs.stanford.edu/people/paulliu}}{paul.liu@stanford.edu}{}{}

\author{Ashkan Norouzi-Fard}{Google Research}{ashkannorouzi@google.com}{https://orcid.org/0000-0002-2336-9826}{}

\author{Ola Svensson}{EPFL}{ola.svensson@epfl.ch}{https://orcid.org/0000-0003-2997-1372}{\flag{snsf_logo}Research supported by the Swiss National Science Foundation project 200021-184656 ``Randomness in Problem Instances and Randomized Algorithms.''}

\author{Rico Zenklusen}{ETH Zurich}{ricoz@ethz.ch}{https://orcid.org/0000-0002-7148-9304}{\flag{eu_emblem}\flag{erc_logo}Research supported in part by Swiss National Science Foundation grant number 200021\_184622. This project has received funding from the European Research Council (ERC) under the European Union's Horizon 2020 research and innovation programme (grant agreement No 817750).}

\acknowledgements{The authors are very grateful to Jan Vondrak for many interesting discussions and comments related to this paper.}
}

\begin{document}
\maketitle
\begin{abstract}
Recent progress in  (semi-)streaming algorithms for monotone submodular function maximization has led to tight results for a simple  cardinality constraint. 
However,  current techniques fail to give a similar understanding for natural generalizations, including matroid constraints. This paper aims at closing this gap.  
For a single matroid of rank $k$ (i.e., any solution has cardinality at most $k$), our main results are:
\begin{itemize}
    \item A single-pass streaming algorithm that uses $\widetilde{O}(k)$ memory and achieves an approximation guarantee of  $0.3178$.
    \item A multi-pass streaming algorithm that uses $\widetilde{O}(k)$ memory and achieves an approximation guarantee of $(1-1/e - \varepsilon)$ by taking a constant (depending on $\eps$) number of passes over the stream. 
\end{itemize}
This improves on the previously best approximation guarantees of $1/4$ and $1/2$ for single-pass and multi-pass streaming algorithms, respectively. In fact, our multi-pass streaming algorithm is \emph{tight} in that any algorithm with a better guarantee than $1/2$ must make several passes through the stream and any algorithm that beats our guarantee of $1-1/e$ must make linearly many passes (as well as an exponential number of value oracle queries). 

Moreover, we show how the approach we use for multi-pass streaming can be further strengthened if the elements of the stream arrive in uniformly random order, implying an improved result for $p$-matchoid constraints.

\end{abstract}




\section{Introduction} \label{sec:introduction}

Submodular function optimization is a classic topic in combinatorial optimization (see, e.g., the book~\cite{schrijver2003combinatorial}). Already in $1978$,  Nemhauser, Wolsey, and Fisher~\cite{nemhauser1978analysis} analyzed a simple greedy algorithm for selecting the most valuable set $S\subseteq V$ of cardinality at most $k$. This algorithm starts with the empty set $S$, and then, for $k$ steps, adds to $S$ the element $u$ with the largest marginal value. 
Assuming the submodular objective function $f$ is also non-negative and monotone,
they showed that the greedy algorithm returns a $(1-1/e)$-approximate solution. Moreover, the approximation guarantee of $1-1/e$ is known to be tight~\cite{feige1998threshold,nemhauser1978best}.

A natural generalization of a cardinality constraint is that of a matroid constraint.
While a matroid constraint is much more expressive than a cardinality constraint, it has often been the case that further algorithmic developments have led to the same or similar guarantees for both types of constraints.
Indeed, for the problem of maximizing a monotone submodular  function subject to a matroid constraint, C{\u{a}}linescu, Chekuri, P{\'{a}}l, and Vondr{\'{a}}k~\cite{calinescu2011maximizing} developed the more advanced continuous greedy method, and showed that it recovers the guarantee $1-1/e$ in this more general setting. 
Since then, other methods, such as local search~\cite{FilmusW14}, have been developed to recover the same optimal approximation guarantee.

 More recently, applications in data science and machine learning~\cite{SubmodularWWW}, with huge problem instances, have motivated the need for space-efficient algorithms, i.e., (semi-)streaming algorithms for (monotone) submodular function maximization. 
 This is now a very active research area, and recent progress has resulted in  a tight understanding of streaming algorithms for maximizing monotone submodular functions with a single cardinality constraint: the optimal approximation guarantee is $1/2$ for single-pass streaming algorithms, and it is possible to recover the guarantee $1-1/e -\varepsilon$ in $O_\varepsilon(1)$ passes. That it is impossible to improve upon $1/2$  in a single pass is due to~\cite{feldman_2020_one-way}, and the first single-pass streaming algorithm to achieve this guarantee is a simple ``threshold'' based algorithm~\cite{badanidiyuru2014streaming} that, intuitively, selects elements with marginal value at least $\OPT/(2k)$. The $(1-1/e-\varepsilon)$ guarantee in $O_\varepsilon(1)$ passes can be obtained using smart implementations of the greedy approach~\cite{badanidiyuru2014fast,abs-1802-06212,mcgregor2019better,MirzasoleimanBK16,Norouzi-FardTMZ18}. 

It is interesting to note that simple greedy and threshold-based algorithms have led to tight results for maximizing a monotone submodular function subject to a cardinality constraint in both the ``offline'' RAM and data stream models.  
However, in contrast to the RAM model, where more advanced algorithmic techniques have generalized these guarantees to much more general constraint families, current techniques fail to give a  similar understanding in the data stream model, both for single-pass and multi-pass streaming algorithms. 
Closing this gap is the motivation for our work.  In particular, 
current results leave open the intriguing possibility to obtain the same guarantees for a matroid constraint as for a cardinality constraint. Our results make significant progress on this question for single-pass streaming algorithms and completely close the gap for multi-pass streaming algorithms.  

\newtoggle{defineRankInTheorems}\toggletrue{defineRankInTheorems}
\begin{restatable}{theorem}{thmMainSingleMatMon} \label{thm:mainSingleMatMon}
    There is a single-pass semi-streaming algorithm for maximizing a non-negative monotone submodular function subject to a matroid constraint of rank $k$ \iftoggle{defineRankInTheorems}{(any solution has cardinality at most $k$) }{}that stores $O(k)$ elements, requires $\widetilde{O}(k)$ additional memory, and achieves an approximation guarantee of $0.3178$.
\end{restatable}
The last theorem improves upon the previous best approximation guarantee of $1/4 = 0.25$~\cite{chakrabarti2015submodular}. Moreover, the techniques are versatile and also yield a single-pass streaming algorithm with an improved approximation guarantee for non-monotone functions  (improving from $0.1715$~\cite{feldman2018do} to $0.1921$).

Our next result is a tight multi-pass guarantee of $1-1/e - \varepsilon$, improving upon the previously best guarantee of $1/2 -\varepsilon$~\cite{HuangTW20}.
\begin{restatable}{theorem}{thmMultipass} \label{thm:multipass}
    For every constant $\varepsilon >0$, there is a multi-pass semi-streaming algorithm for maximizing a non-negative monotone submodular function subject to a matroid constraint of rank $k$\iftoggle{defineRankInTheorems}{ (any solution has cardinality at most $k$)}{} that stores $O(k/\varepsilon)$ elements, makes $O(1/\varepsilon^3)$ many passes, and achieves an approximation guarantee of $1-1/e -\varepsilon$.
\end{restatable}
The result is tight (up to the exact dependency on $\varepsilon$) in the following strong sense: any streaming algorithm with a better approximation guarantee than $1/2$ must make more than one pass~\cite{feldman_2020_one-way}, and any algorithm with a better guarantee than $1-1/e$ must make linearly (in the length of the stream) many passes~\cite{mcgregor2019better} (see \cref{ssc:lower_bound_multipass} for more detail).

The way we obtain~\cref{thm:multipass} is through a rather general and versatile framework based on the ``Accelerated Continuous Greedy'' algorithm of~\cite{badanidiyuru2014fast}, which was designed for the classic (non-streaming) setting.
This allows us to obtain results with an improved number of passes or more general constraints in specific settings.
First, if the elements of the stream arrive in uniformly random order, then we can improve the number of passes as stated below.

\begin{restatable}{theorem}{thmMultipassRandom} \label{thm:multipass_random}
    If the elements arrive in an independently random order in each pass, then for every constant $\varepsilon >0$, there is a multi-pass semi-streaming algorithm for maximizing a non-negative monotone submodular function subject to a matroid constraint of rank $k$\iftoggle{defineRankInTheorems}{ (any solution has cardinality at most $k$)}{} that stores $O(k/\varepsilon)$ elements, makes $O(\varepsilon^{-2} \log \eps^{-1})$ many passes, and achieves an approximation guarantee of $1-1/e -\varepsilon$.
\end{restatable}

Second, also in the uniformly random order model, we can obtain results with even fewer passes, and that also extend to $p$-matchoid constraints, but at the cost of weaker approximation guarantees.
\begin{restatable}{theorem}{thmLocalOptimum} \label{thm:local_optimum}
    If the elements arrive in an independently random order in each pass, then for every constant $\varepsilon >0$, there is a multi-pass semi-streaming algorithm for maximizing a non-negative monotone submodular function subject to a matroid constraint of rank $k$\iftoggle{defineRankInTheorems}{ (any solution has cardinality at most $k$) }{} that stores $O(k)$ elements, makes $O(\log \eps^{-1})$ many passes, and achieves an approximation guarantee of $1/2 -\varepsilon$.
    
    Moreover, if the matroid constraint is replaced with a more general $p$-matchoid constraint, the above still holds except that now the approximation guarantee is $1 / (p + 1) - \varepsilon$ and the number of passes is $O(p^{-1} \log \eps^{-1})$.
\end{restatable}
\togglefalse{defineRankInTheorems}
The $p$-matchoid result of \cref{thm:local_optimum} improves, in the random order model, over an algorithm of~\cite{HuangTW20} that achieves the same approximation factor, but needs $O(p/\eps)$ passes, whereas our algorithm requires a number of passes that only logarithmically depends on $\eps^{-1}$ and decreases (rather than increases) with $p$.
(However, we highlight that the procedure in~\cite{HuangTW20} does not require random arrival order, and obtains its guarantees even in the adversarial arrival model.)



\subsection{Our Technique}\label{sec:ourTechniques}

Before getting into the technical details of our approaches, we provide an overview of the main ingredients behind the techniques we employ.

\paragraph*{Single pass algorithms.} The $4$-approximation single pass algorithm due to Chakrabarti and Kale~\cite{chakrabarti2015submodular} (and later algorithms based on it such as~\cite{chekuri2015streaming,feldman2018do}) maintains an integral solution in the following way. Whenever a new element $u$ arrives, the algorithm considers inserting $u$ into the solution at the expense of some element $u'$ that gets removed from the solution; and this swap is performed if it is beneficial enough. Naturally, the decision to make the swap is a binary decision: we either make the swap or we do not do that. The central new idea in our improved single pass algorithms (\cref{thm:mainSingleMatMon}) is that we make the swap fractional. In other words, we start inserting fractions of $u$ at the expense of fractions of $u'$ (the identity of $u'$ might be different for different fractions of $u$), and we continue to do that as long as the swap is beneficial enough. Since ``beneficial enough'' depends on properties of the current solution, the swapping might stop being beneficial enough before all of $u$ is inserted into the solution, which explains why our fractional swapping does not behave like the integral swapping used by previous algorithms.

While our single pass algorithms are based on the above idea, they are presented in a slightly different way for simplicity of the presentation and analysis. In a nutshell, the differences can be summarized by the following two points.
\begin{itemize}
    \item Instead of maintaining a fractional solution, we maintain multiple sets $A_i$ (for $i \in \bZ$). Membership of an element $u$ in each of these sets corresponds to having a fraction of $1/m$ (for a parameter $m$ of the algorithm) of $u$ in the fractional solution.
    \item We do not remove elements from our fractional solution. Instead, we add new elements to sets $A_i$ with larger and larger $i$ indexes with the implicit view that only fractions corresponding to sets $A_i$ with relatively large indices are considered part of the fractional solution.
\end{itemize}
To make the above points more concrete, we note that the fractional solution is reconstructed from the sets $A_i$ according to the above principles at the very end of the execution of our algorithms. The reconstructed fractional solution is denoted by $s$ in these algorithms.

\paragraph*{Multi-pass algorithms.} Badanidiyuru and Vondr\'{a}k~\cite{badanidiyuru2014fast} described an algorithm called ``Accelerated Continuous Greedy'' that obtains $1 - 1/e - O(\eps)$ approximation (for every $\eps \in (0, 1)$) for maximizing a monotone submodular function subject to a matroid constraint. Even though their algorithm is not a data stream algorithm, it accesses the input only in a well-defined restricted way, namely though a procedure called ``Decreasing-Threshold Procedure''. Originally, this procedure was implemented using a greedy algorithm on an altered objective function. However, we observe that the algorithm of~\cite{badanidiyuru2014fast} can work even if Decreasing-Threshold Procedure is modified to return any local maximum of the same altered objective function. Therefore, to get a multiple pass data stream algorithm, it suffices to design such an algorithm that produces an (approximate) local maximum (or a solution that is as good as such a local maximum); this algorithm can then be used as the implementation of Decreasing-Threshold Procedure. This is the framework we use to get our $(1 - 1/e - \eps)$-approximation algorithms.

To prove \cref{thm:multipass} using the above framework, we show that a known algorithm (a variant of the algorithm of Chakrabarti and Kale~\cite{chakrabarti2015submodular} due to Huang, Thiery, and Ward~\cite{HuangTW20}) can be repurposed to produce an approximate local maximum using $O(\eps^{-2})$ passes, which, when used in Accelerated Continuous Greedy, leads to the claimed $O(\eps^{-3})$ many passes. Similarly, by adapting an algorithm of Shadravan~\cite{S20} working in the random order model, and extending it to multiple passes, we are able to get a solution that is as good as an approximate local maximum in only $O(\eps^{-1} \log \eps^{-1})$ random-order passes, which leads to \cref{thm:multipass_random} when combined with the above framework.

Interestingly, any (approximate) local maximum also has an approximation guarantee of its own (without employing the above framework). This means that the above procedures for producing approximate local maxima can also be viewed as approximation algorithms in their own right, which leads to \cref{thm:local_optimum}.\footnote{Technically, we can also get a result for adversarial order streams in this way, but we omit this result since it is weaker than a known result of~\cite{HuangTW20}.} It is important to note that \cref{thm:local_optimum} uses fewer passes than what is used in the proof of \cref{thm:multipass_random} to get a solution which is at least as good as an approximate local maximum. This discrepancy happens because in  \cref{thm:local_optimum} we only aim for a solution with some approximation ratio $r$, where $r$ is an approximation ratio guaranteed by any approximate local maximum in any instance. In contrast, \cref{thm:multipass_random} needs a solution that is as good as some real approximate local maximum of the particular instance considered.

\subsection{Additional Related Work}

As mentioned above, C{\u{a}}linescu et al.~\cite{calinescu2011maximizing} proposed a $(1-1/e)$-approximation algorithm for maximizing a monotone submodular function subject to a matroid constraint in the offline (RAM) setting, which is known to be tight~\cite{feige1998threshold,nemhauser1978best}. The corresponding problem with a non-monotone objective is not as well understood. A long line of work~\cite{ene2016constrained,feldman2011unified,DBLP:conf/stoc/LeeMNS09} on this problem culminated in a $0.385$-approximation due to Buchbinder and Feldman~\cite{buchbinder2019constrained} and an upper bound by Oveis Gharan and Vondr\'{a}k~\cite{gharan2011submodular} of $0.478$ on the best obtainable approximation ratio.

The first semi-streaming algorithm for maximizing a monotone submodular function subject to a matroid constraint was described by Chakrabarti and Kale~\cite{chakrabarti2015submodular}, who obtained an approximation ratio of $1/4$ for the problem. This remained state-of-the-art prior to this work. However, Chan, Huang, Jiang, Kang, and Tang~\cite{chan2017online} managed to get an improved approximation ratio of $0.3178$ for the special case of a partition matroid in the related preemptive online model. We note that the last approximation ratio is identical to the approximation ratio stated in \cref{thm:mainSingleMatMon}, which points to some similarity that exists between the algorithms (in particular, both use fractional swaps). However, the algorithm of~\cite{chan2017online} is not a semi-streaming algorithm (and moreover, it is tailored to partition matroids). The first semi-streaming algorithm for the non-monotone version of the above problem was obtained by Chekuri, Gupta, and Quanrud~\cite{chekuri2015streaming}, and achieved a $(1/(4 + e) - \eps)\approx 0.1488$-approximation. This was later improved to $0.1715$-approximation by Feldman, Karbasi, and Kazemi~\cite{feldman2018do}.\footnote{Mirzasoleiman et al.~\cite{mirzasoleiman2018streaming} claimed another approximation ratio for the problem (weaker than the one given later by~\cite{feldman2018do}), but some problems were found in their analysis (see \cite{haba2020streaming} for details).}

\paragraph*{Outline of the paper.} In \cref{sec:prelim}, we introduce notations and definitions used throughout this paper. Afterwords, in \cref{sec:singlepass} (and \cref{app:singleMatNonMon}), we present and analyze our single-pass algorithms for maximizing submodular functions subject to a matroid constraint. The framework used to prove \cref{thm:multipass,thm:multipass_random} is presented in detail in \cref{sec:framework}, and in the two sections after it we describe the algorithms for obtaining approximate local maxima (or equally good solutions) necessary for using this framework. Specifically, in \cref{sec:adversarial_local_search} we show how to get such an algorithm for adversarial order streams (leading to \cref{thm:multipass}), and in \cref{sec:random_local_search} we show how to get such an algorithm for random order streams (leading to \cref{thm:multipass_random,thm:local_optimum}).  
It is worth noting that \cref{sec:singlepass} is independent of all the other sections, and therefore, can be skipped by a reader interested in the other parts of this paper.

\section{Preliminaries} \label{sec:prelim}

Recall that we are interested in the problem of maximizing a submodular function subject to a matroid constraint. In \cref{ssc:problem_definition} we give the definitions necessary for formally stating this problem. Then, in \cref{ssc:settings} we define the data stream model in which we study the problem. Finally, in \cref{ssc:additional_definitions} we present some additional notation and definitions that we use.

\subsection{Problem Statement} \label{ssc:problem_definition}

\paragraph*{Submodular Functions.} Given a ground set $\cN$, a \emph{set function} $f\colon 2^\cN \to \bR$ is a function that assigns a numerical value to every subset of $\cN$. Given a set $S \subseteq \cN$ and an element $u \in \cN$, it is useful to denote by $f(u \mid S)$ the marginal contribution of $u$ to $S$ with respect to $f$, i.e., $f(u \mid S) \coloneqq f(S \cup \{u\}) - f(S)$. Similarly, we denote the marginal contribution of a set $T \subseteq \cN$ to $S$ with respect to $f$ by $f(T \mid S) \coloneqq f(S \cup T) - f(S)$.

A set function $f\colon 2^\cN \to \bR$ is called \emph{submodular} if for any two sets $S$ and $T$ such that $S \subseteq T \subseteq \cN$ and any element $u \in \cN \setminus T$ we have
\[
	f(u \mid S) \geq f(u \mid T) \enspace.
\]
Moreover, we say that $f$ is \emph{monotone} if $f(S_1)\leq f(S_2)$ for any sets $S_1\subseteq S_2 \subseteq \cN$, and $f$ is \emph{non-negative} if $f(S) \geq 0$ for every $S \subseteq \cN$.

\paragraph*{Matroids.} A set system is a pair $M = (\mathcal{N}, \mathcal{I})$, where $\mathcal{N}$ is a finite set called the \emph{ground set}, and $\mathcal{I} \subseteq 2^{\mathcal{N}}$ is a collection of subsets of the ground set. We say that a set $S \subseteq \cN$ is \emph{independent} in $M$ if it belongs to $\cI$ (otherwise, we say that it is a \emph{dependent} set); and the rank of the set system $M$ is defined as the maximum size of an independent set in it.
A set system is a \emph{matroid} if it has three properties: i) The empty set is independent, i.e., $\varnothing \in \mathcal{I}$. ii) Every subset of an independent set is independent, i.e., for any $S \subseteq T \subseteq \mathcal{N}$, if $T \in \mathcal{I}$ then $S \in \mathcal{I}$. iii) If $S \in \mathcal{I}$, $T \in \mathcal{I}$ and $|S| < |T|$, then there exists an element $u \in T \setminus S$ such that $S \cup \{u\} \in \mathcal{I}$.\footnote{The last property is often referred to as the \emph{exchange axiom} of matroids.}

A matroid constraint is simply a constraint that allows only sets that are independent in a given matroid. Matroid constraints are of interest because they have a rich combinatorial structure and yet are able to capture many constraints of interest such as cardinality, independence of vectors in a vector space, and being a non-cyclic sub-graph.

\paragraph*{Matchoids and $p$-matchoids.}
The matchoid notion (for the case of $p=2$) was proposed by Jack Edmonds as a common generalization of matching and matroid intersection. Let $M_1 = (\cN_1, \cI_1), M_2 = (\cN_2, \cI_2), \ldots, M_q = (\cN_q, \cI_q)$ be $q$ matroids, and let $\cN = \cN_1 \cup \cdots \cup \cN_q$ and $\cI = \{S \subseteq \cN \mid S \cap \cN_\ell \in \cI_\ell\textrm{ for every integer }1 \leq \ell \leq q\}$. The set system $M = (\cN, \cI)$ is a \emph{$p$-matchoid} if each element $u \in \cN$ is a member of $\cN_\ell$ for at most $p$ indices $\ell \in [q]$. Informally, a $p$-matchoid is an intersection of matroids in which every particular element $u \in \cN$ is affected by at most $p$ matroids. It is easy to see that a $1$-matchoid is simply a matroid, and vice versa. $2$-matchoids are often referred to simply as matchoids (without parameter $p$).


\paragraph*{Problem.} In the \texttt{Submodular Maximization subject to a Matroid Constraint} problem (\SMM), we are given a non-negative\footnote{The assumption of non-negativity is necessary because we are interested in multiplicative approximation guarantees.} submodular function $f\colon 2^\cN \to \nnR$ and a matroid $M = (\cN, \cI)$ over the same ground set. The objective is to find an independent set $S \in \cI$ that maximizes $f$. An important special case of {\SMM} is the \texttt{Monotone Submodular Maximization subject to a Matroid Constraint} problem (\MSMM) in which we are guaranteed that the objective function $f$ is monotone (in addition to being non-negative and submodular). 


\subsection{Data Stream Model} \label{ssc:settings}

In the data stream model, the input appears in a sequential form known as the \emph{input stream}, and the algorithm is allowed to read it only sequentially. In the context of our problem, the input stream consists of the elements of the ground set sorted in either an adversarially chosen order or a uniformly random order, and the algorithm is allowed to read the elements from the stream only in this order. Often the algorithm is allowed to read the input stream only once (such algorithms are called \emph{single-pass} algorithms), but in other cases it makes sense to allow the algorithm to read the input stream multiple times---each such reading is called a \emph{pass}. The order of the elements in each pass might be different; in particular, when the order is random, we assume that it is chosen independently for each pass.

A trivial way to deal with the restrictions of the data stream model is to store the entire input stream in the memory of the algorithm. However, we are often interested in a stream carrying too much data for this to be possible.
Thus, the goal in this model is to find a high quality solution while using significantly less memory than what is necessary for storing the input stream.
The gold standard are algorithms that use memory of size nearly linear in the maximum possible size of an output; such algorithms are called \emph{semi-streaming} algorithms.\footnote{The similar term \emph{streaming} algorithms often refers to algorithms whose space complexity is poly-logarithmic in the parameters of their input. Such algorithms are irrelevant for the problem we consider because they do not have enough space even for storing the output of the algorithm.} For {\SMM} and {\MSMM}, this implies that a semi-streaming algorithm is a data stream algorithm that uses $O(k \log^{O(1)} |\cN|)$ space, where $k$ is the rank of the matroid constraint.

The description of submodular functions and matroids can be exponential in the size of their ground sets, and therefore, it is important to define the way in which an algorithm may access them. We make the standard assumption that the algorithm has two oracles: a \emph{value oracle} and an \emph{independence oracle} which, given a set $S \subseteq \cN$ of elements that are explicitly \emph{stored} in the memory of the algorithm, returns the value of $f(S)$ and an indicator whether $S \in \cI$, respectively.

\subsection{Additional Notation and Definitions} \label{ssc:additional_definitions}

\paragraph*{Multilinear Extension.} A set function $f\colon 2^\cN \to \bR$ assigns values only to subsets of $\cN$. If we think of a set $S$ as equivalent to its characteristic vector $\characteristic_S$ (a vector in $\{0, 1\}^\cN$ that has a value of $1$ in every coordinate $u \in S$ and a value of $0$ in the other coordinates), then we can view $f$ as a function over the integral vectors in $[0, 1]^\cN$. It is often useful to extend $f$ to general vectors in $[0, 1]^\cN$. There are multiple natural ways to do that. However, in this paper, we only need the multilinear extension $F$. Given a vector $x \in [0, 1]^\cN$, let $\RSet(x)$ denote a random subset of $\cN$ that includes each element $u \in \cN$ with probability $x_u$, independently. Then,
\[
	F(x)
	=
	\bE[f(\RSet(x))]
	=
	\sum_{S \subseteq \cN} \left[f(S) \cdot \prod_{u \in S} x_u \cdot \prod_{u \not \in S} (1 - x_u)\right]
	\enspace.
\]
One can observe that, as is implied by its name, the multilinear extension is a multilinear function. This implies that, for every vector $x \in [0, 1]^\cN$, the partial derivative $\frac{\partial F}{\partial x_u}(x)$ is equal to $F(x + (1 - x_u) \cdot \characteristic_u) - F(x - x_u \cdot \characteristic_u)$. Note that in the last expression we have used $\characteristic_u$ as a shorthand for $\characteristic_{\{u\}}$. We often also use $\partial_u F(x)$ as a shorthand for $\frac{\partial F}{\partial x_u}(x)$. When $f$ is submodular, its multilinear extension $F$ is known to be concave along non-negative directions~\cite{calinescu2011maximizing}.

\paragraph*{General Notation.} Given a set $S \subseteq \cN$ and an element $u \in \cN$, we denote by $S + u$ and $S - u$ the expressions $S \cup \{u\}$ and $S \setminus \{u\}$, respectively. Additionally, given two vectors $x, y \in [0, 1]^\cN$, we denote by $x \vee y$ and $x \wedge y$ the coordinate-wise maximum and minimum operations, respectively.

\paragraph*{Additional Definitions from Matroid Theory.} Matroid theory is very extensive, and we refer the reader to~\cite{schrijver2003combinatorial} for a more complete coverage of it. Here, we give only a few basic definitions from this theory that we employ below.
Given a matroid $M = (\cN, \cI)$, a set $S \subseteq \cN$ is called \emph{base} if it is an independent set that is maximal with respect to inclusion (i.e., every super-set of $S$ is dependent), and it is called \emph{cycle} if it is a dependent set that is minimal with respect to inclusion (i.e., every subset of $S$ is independent). An element $u \in \cN$ is called a self-loop if $\{u\}$ is a cycle. Notice that such elements cannot appear in any feasible solution for either {\SMM} or {\MSMM}, and therefore, one can assume without loss of generality that there are no self-loops in the ground set.

The rank of a set $S \subseteq \cN$, denoted by $\rank_M(S)$, is the maximum size of an independent set $T \in \cI$ which is a subset of $S$. The subscript $M$ is omitted when it is clear from the context. We also note that $\rank_M(\cN)$ is exactly the rank of the matroid $M$ (i.e., the maximum size of an independent set in $M$), and therefore, it is customary to define $\rank(M) = \rank_M(\cN)$. We say that a set $S \subseteq \cN$ spans an element $u \in \cN$ if adding $u$ to $S$ does not increase the rank of the set $S$, i.e., $\rank(S) = \rank(S + u)$---observe the analogy between this definition and being spanned in a vector space. Furthermore, we denote by $\Span_M(S) \coloneqq \{u \in \cN \mid \rank(S) = \rank(S + u)\}$ the set of elements that are spanned by $S$. Again, the subscript $M$ is dropped when it is clear from the context.

\section{Single-Pass Algorithm}
\label{sec:singlepass}

In this section, we present our single-pass algorithm for the \texttt{Monotone Submodular Maximi\-zation subject to a Matroid Constraint} problem (\MSMM). 
The properties of the algorithm we present are given by the following theorem.
\thmMainSingleMatMon*

Our algorithm can be extended to the case in which the objective function is non-monotone (i.e., the {\SMM} problem) at the cost of obtaining a lower approximation factor, yielding the following theorem. However, for the sake of concentrating on our main new ideas, we devote this section to the proof of \cref{thm:mainSingleMatMon} and defer the proof of \cref{thm:mainSinglaMatNonMon} to \cref{app:singleMatNonMon}.
\begin{theorem}\label{thm:mainSinglaMatNonMon}
There is a single-pass semi-streaming algorithm for maximizing a non-negative (not necessarily monotone) submodular function subject to a matroid constraint of rank $k$ that stores $O(k)$ elements, requires $\widetilde{O}(k)$ additional memory, and achieves an approximation guarantee of $0.1921$.
\end{theorem}

Throughout this section, we denote by $P_M \coloneqq \{x\in \mathbb{R}_{\geq 0}^\cN \colon x(S) \leq \rank(S) \;\forall S\subseteq \cN\}$ the matroid polytope of $M$.

\subsection{The algorithm} \label{sec:algorithm}

The algorithm we use to prove \cref{thm:mainSingleMatMon} appears as \cref{alg:singleMatMonotone}. This algorithm gets a parameter $\eps > 0$ and starts by initializing a constant $\alpha$ to be approximately the single positive value obeying $\alpha + 2 = e^{\alpha}$. We later prove that the approximation ratio guaranteed by the algorithm is at least $\frac{1}{\alpha + 2} - \eps$, which is better than the approximation ratio stated in \cref{thm:mainSingleMatMon} for a small enough $\eps$. After setting the value of $\alpha$, \cref{alg:singleMatMonotone} defines some additional constants $m$, $c$, and $L$ using $\eps$ and $\alpha$. We leave these variables representing different constants as such in the procedure and analysis, which allows for obtaining a better understanding later on of why these values are optimal for our analysis. We also note that, as stated, \cref{alg:singleMatMonotone} is efficient (i.e., runs in polynomial time) only if the multilinear extension and its partial derivatives can be efficiently evaluated. If that cannot be done, then one has to approximate $F$ and its derivatives using Monte-Carlo simulation, which is standard practice (see, for example, \cite{calinescu2011maximizing}).
We omit the details to keep the presentation simple, but we note that, as in other applications of this standard technique, the incurred error can easily be kept negligible, and therefore, does not affect the guarantee stated in \cref{thm:mainSingleMatMon}.

\cref{alg:singleMatMonotone} uses sets $A_i$ and vectors $a_i\in [0,1]^\cN$ for certain indices $i\in \mathbb{Z}$.
Throughout the algorithm, we only consider finitely many indices $i\in \mathbb{Z}$.
However, we do not know upfront which indices within $\mathbb{Z}$ we will use.
To simplify the presentation, we therefore use the convention that whenever the algorithm uses for the first time a set $A_i$ or vector $a_i$, then $A_i$ is initialized to be $\varnothing$ and $a_i$ is initialized to be the zero vector.
The largest index ever used in the algorithm is $q$, which is computed toward the end of the algorithm at Line~\ref{algline:define_q}.

For each $i\in \mathbb{Z}$, the set $A_i$ is an independent set consisting of elements $u$ that already arrived and for which the marginal increase with respect to a reference vector $a$ (at the moment when $u$ arrives) is at least $c^i$.
More precisely, whenever a new element $u\in \cN$ arrives and its marginal return $\partial_u F(a)$ exceeds $c^i$ for an index $i\in \mathbb{Z}$ in a relevant range, then we add $u$ to $A_i$ if $A_i + u$ remains independent.
When adding $u$ to $A_i$, we also increase the $u$-entry of the vector $a_i$ by $\frac{c^i}{m\cdot \partial_u F(a)}$.
The vector $a$ built up during the algorithm has two key properties.
First, its multilinear value approximates $f(\OPT)$ up to a constant factor.
Second, one can derive from the sets $A_i$ a vector $s$ (see \cref{alg:singleMatMonotone}) such that $F(s)$ is close to $F(a)$ and $s$ is contained in the matroid polytope $P_\cM$.

Whenever an element $u\in \cN$ arrives, the algorithm first computes the largest index $i(u)\in \mathbb{Z}$ fulfilling $c^{i(u)} \leq \partial_u F(a)$. It then updates sets $A_i$ and vectors $a_i$ for indices $i \leq i(u)$.
Purely conceptually, the output of the algorithm would have the desired guarantees even if all infinitely many indices below $i(u)$ where updated. 
However, to obtain an algorithm running in finite (actually even polynomial) time and linear memory, we do not consider indices below $\max\{b, i(u) - \rank(M) - L\}$ in the update step.
Capping the considered indices like this has only a very minor impact in the analysis because the contribution of the vectors $a_i$ to the multilinear extension value of the vector $a$ is geometrically decreasing with decreasing index $i$.

In the algorithm, and also the analysis that follows, we sometimes use sums over indices that go up to $\infty$. However, whenever this happens, beyond some finite index, all terms are zero. Hence, such sums are well defined.

\begingroup
\begin{algorithm}[ht]
\caption{Single-Pass Semi-Streaming Algorithm for {\MSMM}} \label{alg:singleMatMonotone}
\begin{algorithmic}[1]
\State Set $\alpha = 1.1462$, $m=\left\lceil\frac{3 \alpha}{\eps}\right\rceil$, $c=\frac{m}{m-\alpha}$, and $L = \left\lceil \log_c(\frac{2 c}{\eps(c-1)}) \right\rceil$.
\State Set $a = 0 \in [0, 1]^\cN$ to be the zero vector, and let $b=-\infty$.

\For{every element arriving $u \in \cN$, if $\partial_u F(a) > 0$}\label{algline:loopOverU} 
	\State Let $i(u) = \left\lfloor \log_c(\partial_u F(a)) \right\rfloor$.\label{algline:index_i_u}
	\Comment{%
	\raggedright Thus, $i(u)$ is largest index $i\in \mathbb{Z}$ with $c^i \leq \partial_u F(a)$.%
	}
	\For{$i = \max\{b, i(u) - \rank(M) - L\}$ \textbf{to} $i(u)$}\label{algline:loopOverI}
		\If{$A_i + u \in \cI$}
			\State $A_i \gets A_i + u$.
			\State $a_i \gets a_i + \frac{c^i}{m\cdot \partial_u F(a)} \characteristic_u$.\label{algline:modify_ai}
		\EndIf
	\EndFor

\State Set $b \gets h - L$, where $h$ is largest index $i\in \mathbb{Z}$ satisfying $\sum_{j=i}^\infty |A_j| \geq \rank(M)$.

\State $a \gets \sum_{i=b}^{\infty} a_i$.\label{algline:updateA}

\State Delete from memory all sets $A_i$ and vectors $a_i$ with $i\in \mathbb{Z}_{<b}$.\label{algline:reduceMem}

\EndFor

\State Set $S_k \gets \varnothing$ for $k \in \{0,\ldots, m-1\}$.

\State Let $q$ be largest index $i\in \mathbb{Z}$ with $A_i\neq \varnothing$.\label{algline:define_q}

\For{$i = q$ \textbf{to} $b$ (stepping down by $1$ at each iteration)}\label{algline:build_Sk}
	\While{$\exists u \in A_i \setminus S_{(i \bmod m)}$ with $S_{(i \bmod m)} + u\in \cI$}
		\State $S_{(i \bmod m)} \gets S_{(i \bmod m)} + u$.
	\EndWhile
\EndFor
\State \Return{a rounding $R\in \mathcal{I}$ of the fractional solution $s\coloneqq \frac{1}{m}\sum_{k=0}^{m-1} \characteristic_{S_k}$ with $f(R) \geq F(s)$}.\label{algline:returnMonotSubmodMat}
\end{algorithmic}
\end{algorithm}
\endgroup

Finally, we provide details on the return statement in Line~\ref{algline:returnMonotSubmodMat} of the algorithm.
This statement is based on a fact stated in~\cite{calinescu2011maximizing}, namely that a point in the matroid polytope can be rounded losslessly to an independent set.
More formally, given any point $y\in P_\cM$ in the matroid polytope, there is an independent set $I\in \cI$ with $f(I) \geq F(y)$.
Moreover, assuming that the multilinear extension $F$ can be evaluated efficiently, such an independent set $I$ can be computed efficiently.
As before, if one is only given a value oracle for $f$, then the exact evaluation of $F$ can be replaced by a strong estimate obtained through Monte-Carlo sampling, leading to a randomized algorithm to round $y$ to an independent set $I$ with $f(I) \geq (1-\delta)F(y)$ for an arbitrarily small constant $\delta > 0$.

We highlight that we can assume in what follows that there is at least one element $u\in \cN$ which gets considered in the for-loop on Line~\ref{algline:loopOverU}, i.e., it fulfills $\partial_u F(a) > 0$ when appearing in the for-loop.
Note that if this does not happen, then we are in a trivial special case where $a$ remains the zero vector and $\partial_u F(a) = 0$ for all $u\in \cN$, which corresponds to $f(\cN) = f(\varnothing)$.
In this case, all sets $S_k$ for $k\in \{0,\ldots, m-1\}$ are empty, which implies that $s$ is the zero vector, and one can simply return $R=\varnothing$, which fulfills $f(R) \geq F(s)$, and is even a global maximizer of $f(S)$ over all sets $S\subseteq \cN$.

\subsection{Analysis of \texorpdfstring{\cref{alg:singleMatMonotone}}{Algorithm~\ref{alg:singleMatMonotone}}}

We now show that \cref{alg:singleMatMonotone} implies \cref{thm:mainSingleMatMon}.
Let $\eps \in (0,1]$ in what follows.
As mentioned, we sometimes restrict the considered index range for $i$ in the algorithm to make sure that the algorithm has a finite running time and only uses limited memory.
This happens in particular in Line~\ref{algline:updateA} when updating $a$, where we only consider indices starting from $b$.
However, for the analysis, it is convenient to look at the vector $\aall = \sum_{i=-\infty}^\infty \overline{a}_i$ obtained without this lower bound, where $\overline{a}_i$ is the vector $a_i$ when the algorithm terminates.
In the definition of $\aall$, we also consider indices $i\in \mathbb{Z}$ together with corresponding vectors $\overline{a}_i$ that have been removed from memory in Line~\ref{algline:reduceMem}.
Here, the vector $\overline{a}_i$ is simply the last vector $a_i$ before it got removed from memory in Line~\ref{algline:reduceMem}.
Similarly, we let $\overline{A}_i\subseteq \cN$ be the set $A_i$ at the end of the algorithm or, in case $A_i$ got removed from memory at some point, $\overline{A}_i$ is the set $A_i$ right before it got removed from memory. 

Note that because the coordinates of the vectors $a_i$ never decrease throughout the algorithm, every vector $a$ encountered throughout \cref{alg:singleMatMonotone} is upper bounded, coordinate-wise, by $\aall$.
In the following, we compare both the value of an optimal solution and the value $f(R) \geq F(s)$ of the returned set to $F(\aall)$.
We start by making sure that the different steps of the algorithm are well defined.
For this, we first show that $\aall$, and therefore also any vector $a$ encountered through \cref{alg:singleMatMonotone}, is contained in the box $[0,1]^\cN$, which implies that the computations of partial derivatives $\partial_u F(a)$ are well defined.

\begin{observation} \label{obs:overline_a_is_in_box}
It holds that $\aall\in [0,1]^\cN$. Consequently, throughout the algorithm, the vector $a$ is also contained in $[0,1]^\cN$.
\end{observation}
\begin{proof}
Consider an element $u\in \cN$ and the moment when $u$ was considered in the for-loop at Line~\ref{algline:loopOverU} of \cref{alg:singleMatMonotone}.
Let $i(u)$ be the index computed at Line~\ref{algline:index_i_u} of the algorithm.
Hence, for the vector $a$ at that moment we have $c^{i(u)} \leq \partial_u F(a)$.
Thus,
\begin{align*}
\aall(u) \leq \sum_{j=-\infty}^{i(u)} \frac{c^j}{m\cdot \partial_u F(a)}
 \leq \frac{1}{m} \sum_{j=-\infty}^0 c^j
 =\frac{1}{m}\frac{c}{c-1}
 =\frac{1}{\alpha}
 \leq 1\enspace,
\end{align*}
where the second inequality follows from $c^{i(u)} \leq \partial_u F(a)$, and the second equality holds by the definition of $c$, i.e., $c=\frac{m}{m-\alpha}$.
\end{proof}

Moreover, we highlight that the fractional point $s$ rounded at the end of \cref{alg:singleMatMonotone} at Line~\ref{algline:returnMonotSubmodMat} is indeed in the matroid polytope $P_M$. This holds because it is a convex combination of the sets $S_k$ for $k\in \{0,\ldots m-1\}$, each of which is an independent set by construction. Hence, the rounding performed in Line~\ref{algline:returnMonotSubmodMat} is indeed possible, as discussed.

We now bound the memory used by the algorithm. Note that, for any constant $\eps$ (which implies that $c$ is also a constant), the guarantee in the next lemma becomes $O(\rank(M))$, which is the guarantee we need in order to prove \cref{thm:mainSingleMatMon}. One can also observe that \cref{alg:singleMatMonotone} stores one non-zero entry in its $a_i$ vectors for every element stored in the sets $A_i$; and thus, the next lemma also implies that \cref{alg:singleMatMonotone} is a semi-streaming algorithm using space $\widetilde{O}(\rank(M))$.
\begin{lemma} \label{lem:elements_in_memory_raw}
At any point in time, the sum of the cardinalities of all sets $A_i$ that \cref{alg:singleMatMonotone} has in memory is $O(L \cdot \rank(M)) = O\left(\frac{\log\left(\frac{2c}{\eps (c - 1)}\right)}{\log c} \rank(M)\right)$.
\end{lemma}
\begin{proof}
It suffices to bound the number of elements $\sum_{i=b}^\infty |A_i|$ after Line~\ref{algline:reduceMem}.
Indeed, we never have more than that many elements in memory plus the number of elements added in a single iteration of the for-loop at Line~\ref{algline:loopOverU}, which is at most $\rank(M)+ L = O(L \cdot \rank(M))$. 
Hence, consider the state of the algorithm at any moment right after the executing of Line~\ref{algline:reduceMem}.
We have
\begin{align*}
\sum_{i=b}^{\infty} |A_i| =     \sum_{i=b}^{b+L} |A_i| + \sum_{i=b+L+1}^{\infty} |A_i|
      									  < (L+1) \rank(\cM) + \rank(\cM)
											    = O(L \cdot \rank(\cM))\enspace,
\end{align*}
where the inequality follows from the fact that the first sum has $L+1$ terms, each is the cardinality of an independent set, which is upper bounded by $\rank(M)$; moreover, the second term in the sum is strictly less than $\rank(M)$ by the definition of $h$ in \cref{alg:singleMatMonotone} (note that $h = b+L$).
\end{proof}

We now start to relate the different relevant quantities to $F(\aall)$. We start by upper bounding the value of $F(\aall)$ as a function of the sets $\overline{A}_i$.
\begin{lemma}\label{lem:upperBoundFAall}
\begin{equation*}
F(\aall) \leq f(\varnothing) + \frac{1}{m} \sum_{i\in \mathbb{Z}} |\overline{A}_i|\cdot c^i\enspace.
\end{equation*}
\end{lemma}
\begin{proof}
One can think of the vector $\aall$ as being constructed iteratively starting with the zero vector $w=0$ as follows. Whenever \cref{alg:singleMatMonotone} is at Line~\ref{algline:modify_ai}, we update $w$ by $w \gets w + \frac{c^i}{m\cdot \partial F(a)}\characteristic_u$, where $a\in [0,1]^\cN$ is the current vector $a$ of the algorithm at that moment in the execution. Note that we have $w\geq a$ because $a=\sum_{j=b}^\infty a_j$, for the current value of $b$ and the current vectors $a_j$, whereas $w=\sum_{j\in \mathbb{Z}}a_j$. Hence, by submodularity of $f$, we have that the increase of $F(w)$ in this iteration is upper bounded by
\begin{equation*}
F\left(q+\frac{c^i}{m\cdot \partial_u F(a)}\characteristic_u\right) - F(q) \leq F\left(a + \frac{c^i}{m\cdot \partial_u F(a)}\characteristic_u\right) - F(a) = \frac{c^i}{m} \enspace.
\end{equation*}
Hence, the total change in $F(w)$ starting from $F(0)=f(\varnothing)$ to $F(\aall)$ is therefore obtained by summing the above left-hand side over all occurrences when algorithm is at Line~\ref{algline:modify_ai}, which leads to
\begin{equation*}
F(\aall) - F(0) \leq \frac{1}{m} \sum_{i\in \mathbb{Z}} |\overline{A}_i|\cdot c^i\enspace,
\end{equation*}
thus completing the proof.
\end{proof}

Let $\overline{b}$ be the value of $b$ at the end of the algorithm.
A key difference between the fractional point $s$, which is constructed during the algorithm, and the point $\aall$, is that sets $\overline{A}_i$ for indices below $\overline{b}$ have an impact on the value of $F(\aall)$ (but not on $F(s)$), as reflected in the upper bound on $F(\aall)$ in \cref{lem:upperBoundFAall}.
The following lemma shows that this difference in index range is essentially negligible because the impact of the sets $\overline{A}_i$ in these bounds decreases exponentially fast with decreasing index $i$.
\begin{lemma}\label{lem:lowLevelsAreNegligible}
\begin{equation*}
\frac{1}{m} \sum_{i=-\infty}^{\overline{b}-1} c^i \cdot |\overline{A}_i|
  \leq \frac{c^{\overline{b}}}{m(c-1)} \rank(M)
  \leq \frac{\eps}{2c} \cdot \frac{1}{m} \sum_{i=\overline{b}}^q c^i \cdot |\overline{A}_i| \enspace.
\end{equation*}
\end{lemma}
\begin{proof}
The first inequality of the statement follows from $|\overline{A}_i|\leq \rank(M)$ for $i\in \mathbb{Z}$, which holds because $\overline{A}_i\in \cI$.
The second one follows from
\begin{equation}\label{eq:simpleLowerBoundFABar}
\frac{1}{m} \sum_{i= \overline{b}}^q c^i \cdot |\overline{A}_j| 
 \geq \frac{1}{m} \sum_{i= \overline{b}+L}^q c^i \cdot |\overline{A}_j| 
 \geq \frac{1}{m} \rank(M) \cdot c^{\overline{b}+L}
 \geq \frac{1}{m} \rank(M) \cdot c^{\overline{b}} \frac{2c}{\eps (c-1)}\enspace,
\end{equation}
where the second inequality follows by the fact that $\overline{b}+L$ is the value of $h$ at the end of the algorithm, which fulfills by definition $\sum_{i=h}^\infty |\overline{A_j}| \geq \rank(M)$, and the third inequality follows by our definition of $L$.
\end{proof}

Combining \cref{lem:lowLevelsAreNegligible} with \cref{lem:upperBoundFAall} now leads to the following lower bound on $F(\aall)$, described only in terms of sets $|\overline{A}_i|$ that have not been deleted from memory when the algorithm terminates.
\begin{corollary}\label{cor:boundFAallInAi}
\begin{equation*}
F(\aall) - f(\varnothing) \leq \left(1+\frac{\eps}{2c}\right) \cdot \frac{1}{m} \sum_{i=\overline{b}}^q c^i \cdot |\overline{A}_i|\enspace.
\end{equation*}
\end{corollary}
\begin{proof}
The statement follows from
\begin{equation*}
\left(1 + \frac{\eps}{2c}\right)\cdot \frac{1}{m} \sum_{i=\overline{b}}^q c^i \cdot |\overline{A}_i|
  \geq \frac{1}{m} \sum_{i=-\infty}^q c^i \cdot |\overline{A}_i|
  \geq F(\aall) - f(\varnothing)\enspace,
\end{equation*}
where the first inequality is due to \cref{lem:lowLevelsAreNegligible}, and the second one follows from \cref{lem:upperBoundFAall}.
\end{proof}

Before relating $F(s)$ to $F(\aall)$, we need the following structural property on the sets $\overline{A}_i$, which will be exploited to show that the sets $S_k$, chosen at the end of the algorithm, lead to a point $s$ of high multilinear value.
\begin{lemma}\label{lem:Ai_spanned_by_previous}
\begin{equation*}
\overline{A}_i \subseteq \Span\left(\overline{A}_{i-1}\right) \quad \forall i\in \{\overline{b}+1, \overline{b}+2,\ldots, q\}\enspace.
\end{equation*}
\end{lemma}
\begin{proof}
Let $u\in \overline{A}_i$, and we show the statement by proving that $u\in \Span(\overline{A}_{i-1})$.
Consider the state of \cref{alg:singleMatMonotone} when it performs the for-loop at Line~\ref{algline:index_i_u} when the outer for-loop is considering the element $u$ (this is the for-loop that adds the element $u$ to sets $A_j$).
Note that we have
\begin{equation*}
\overline{b} \geq \max\{b, i(u) - \rank(M) - L\}\enspace,
\end{equation*}
due to the following. We clearly have $\overline{b} \geq b$ because the value of $b$ is non-decreasing throughout the algorithm.
Moreover, if $i(u) - \rank(M) - L \geq b$, then after the execution of the for-loop at Line~\ref{algline:index_i_u}, we have $A_j\neq \varnothing$ for each $j\in \{i(u) - \rank(M) - L, \ldots, i(u)\}$.
Hence, right after this execution of the for-loop, the value of $b$ will be increased to at least $i(u) - \rank(M) - L$.

Thus, because $u$ got added to $A_i$ for some $i\in \mathbb{Z}_{> \overline{b}}$, the algorithm will also add $u$ to $A_{i-1}$ if $A_{i-1} + u \in \cI$.
Hence, after the execution of this for-loop, we have $u\in \Span(A_{i-1})$.
Finally, because $\overline{A}_{i-1} \supseteq A_{i-1}$, we also have $u\in \Span(\overline{A}_{i-1})$.
\end{proof}

We are now ready to lower bound the value of $F(s)$ in terms of $F(\aall)$.
\begin{lemma}\label{lem:lowerBoundFsInFAall}
\begin{equation*}
F(s) \geq f(\varnothing) + \frac{1}{m} (1-c^{-m}) \sum_{i=\overline{b}}^q c^i \cdot |\overline{A}_i|
     \geq \left(1-c^{-m} - \frac{\eps}{2c}\right) \cdot F(\aall)\enspace.
\end{equation*}
\end{lemma}
\begin{proof}
For $k\in \{0,\ldots, m-1\}$, we partition $S_k$ into
\begin{equation*}
S_k = S_k^{\overline{b}} \cup S_k^{\overline{b}+1} \cup \ldots \cup S_k^q\enspace,
\end{equation*}
where $S_k^i$ are the elements that got added to $S_k$ in iteration $i$ of Line~\ref{algline:build_Sk}.
Note that because $S_k$ only gets updated in every $m$-th iteration, we have $S_k^i = \varnothing$ for any $i\not\equiv k \pmod{m}$.
Moreover, we have
\begin{equation}\label{eq:lowerBoundSki}
|S_k^i| \geq |\overline{A}_i| - |\overline{A}_{i+m}| \qquad \forall i\in \mathbb{Z} \text{ with } \overline{b} \leq i \leq q \text{ and } i\equiv k \pmod{m}
\end{equation}
because of the following.
We recall that the sets $S_k$ are constructed by adding elements from sets $\overline{A}_j$ from higher indices $j$ to lower ones.
Thus, when elements of $\overline{A}_i$ are considered to be added to $S_k$, the current set $S_k$ only contains elements from sets $A_j$ with $j \geq i + m$ (recall that only elements from every $m$-th set $A_j$ can be added to $S_k$). 
However, by \cref{lem:Ai_spanned_by_previous}, we have that all those elements are spanned by $\overline{A}_{i+m}$.
Hence, when elements of $\overline{A}_i$ are considered to be added to $S_k$, the set $S_k$ has at most $\rank(\overline{A}_{i+m}) = |\overline{A}_{i+m}|$ many elements since $S_k\in \cI$ by construction.
Moreover, when elements of the set $\overline{A}_{i}$ are added to $S_k$, this is done in a greedy way, which implies that the size of $S_k$ after adding elements from $\overline{A}_i$ will be equal to $\rank(\overline{A}_i) = |\overline{A_i}|$.
This implies \cref{eq:lowerBoundSki}.

The desired relation now follow from
{\allowdisplaybreaks\begin{align*}
F(s) &\geq f(\varnothing) + \frac{1}{m} \sum_{i=\overline{b}}^q c^i \sum_{k=0}^{m-1} |S_k^i| \\
     &\geq f(\varnothing) + \frac{1}{m} \sum_{i=\overline{b}}^q c^i \cdot \left(|\overline{A}_i| - |\overline{A}_{i+m}|\right) \\
     &\geq f(\varnothing) + \frac{1}{m} (1 - c^{-m}) \sum_{i=\overline{b}}^q c^i \cdot |\overline{A}_i| \\
     &\geq f(\varnothing) + \frac{1}{m} (1 - c^{-m}) \left(1+\frac{\eps}{2c}\right)^{-1} (F(\aall) - f(\varnothing) \\
     &\geq \frac{1}{m} (1 - c^{-m}) \left(1-\frac{\eps}{2c}\right) F(\aall)\enspace, \\
     &\geq \frac{1}{m} \left(1 - c^{-m} - \frac{\eps}{2c}\right) F(\aall)\enspace, \\
\end{align*}}
where the first inequality follows from a reasoning analogous to the one used in the proof of \cref{lem:upperBoundFAall},\footnote{%
More precisely, we can think of $s$ as being constructed iteratively starting from $w=0$.
Whenever the algorithm adds an element $u\in \cN$ to some set $A_i$ with $i\in \{\overline{b},\ldots, q\}$, then, if $u$ is also part of the set $S_k^i$ for $k\in \{0,\ldots, m-1\}$ with $i\equiv k \pmod{m}$, we update $w$ by setting it to $w+\frac{c^i}{m \cdot \partial_u F(a)}\characteristic_u$.
The increase $F(w+\frac{c^i}{m \cdot \partial_u F(a)}\characteristic_u) - F(w)$ is at least as big as $F(a+\frac{c^i}{m\cdot \partial_u F(a)}\characteristic_u) - F(a)$, which is $\frac{c^i}{m}$.} the second one is due to \cref{eq:lowerBoundSki}, and the fourth one uses \cref{cor:boundFAallInAi}\enspace.
\end{proof}

Let $\OPT$ be an arbitrary (but fixed) optimal solution for our problem. To relate $f(\OPT)$ to $F(\aall)$, we analyze by how much $f(\OPT)$ can be bigger than $F(\aall)$.
This difference can be bounded through the derivatives $\partial_u F(\aall)$, which we analyze first.
To this end, for any $u\in \cN$, we denote by $\ell(u)$ the largest index $i\in \mathbb{Z}$ such that $u\in \Span(\overline{A}_{i})$. If no such index exists, we set $\ell(u) = -\infty$.

\begin{observation}\label{obs:boundDerivativesFAall}
\begin{equation*}
\partial_u F(\aall) \leq c^{\ell(u)+1} \quad \forall u\in \cN \enspace.
\end{equation*}
\end{observation}
\begin{proof}
Because $u\not\in \Span(\overline{A}_{\ell(u)+1})$, this implies that $u$ did not get added to the set $A_{\ell(u)+1}$ in \cref{alg:singleMatMonotone}, even though $A_{\ell(u)+1} + u \in \cI$, which holds because $A_{\ell(u)+1} + u \subseteq \overline{A}_{\ell(u)+1} + u \in \cI$. Hence, when $u$ got considered in Line~\ref{algline:loopOverU} of \cref{alg:singleMatMonotone}, we had $\partial_u F(a) < c^{\ell(u)+1}$. Finally, by submodularity of $f$ and because $a \leq \aall$ (coordinate-wise), we have $\partial_u F(\aall) \leq \partial_u F(a) \leq c^{\ell(u)+1}$.
\end{proof}

We are now ready to bound the difference between $f(\OPT)$ and $F(\aall)$.
\cref{lem:diffOptToAall} is the first statement in our analysis that exploits monotonicity of $f$.
\begin{lemma}\label{lem:diffOptToAall}
\begin{equation*}
f(\OPT) - F(\aall) \leq \sum_{u\in \OPT} c^{\ell(u) + 1}\enspace.
\end{equation*}
\end{lemma}
\begin{proof}
The result follows from
\begin{align*}
f(\OPT) - F(\aall)
  &\leq F(\aall \vee \characteristic_{\OPT}) - F(\aall)                 \\
  &\leq \nabla F(\aall)^T ((\aall \vee \characteristic_{\OPT}) - \aall) \\
  &\leq \nabla F(\aall)^T \characteristic_{\OPT}                        \\
  &= \sum_{u\in \OPT} \partial_u F(\aall)                               \\
  &\leq \sum_{u\in \OPT} c^{\ell(u)+1}\enspace,
\end{align*}
where the first inequality follows from monotonicity of $F$, the second one because $F$ is concave along non-negative directions, the third one uses again monotonicity of $F$ which implies $\nabla F(\aall) \geq 0$, and the last one follows from \cref{obs:boundDerivativesFAall}.
\end{proof}

The following lemma allows us to express the bound on the difference between $f(\OPT)$ and $F(\aall)$ in terms of $F(\aall)$, which, combined with the previously derived results, will later allow us to compare $F(s)$ to $f(\OPT)$ via the quantity $F(\aall)$.
\begin{lemma}\label{lem:upperBoundSumCOpt}
\begin{equation*}
\sum_{u\in \OPT} c^{\ell(u) + 1} \leq (c-1) \left(1 + \frac{\eps}{2c}\right) \frac{m}{1-c^{-m}} F(s) \enspace.
\end{equation*}
\end{lemma}
\begin{proof}

We start by expanding the left-hand side of the inequality to be shown:
\begin{equation}\label{eq:rewriteSumCOpt}
\begin{aligned}
\sum_{u\in \OPT} &c^{\ell(u) + 1}
 = c \cdot \sum_{u\in \OPT} c^{\ell(u)} \\
 &= c \cdot \sum_{i\in \mathbb{Z}} c^i \cdot |\{u\in \OPT \colon \ell(u) = i\}| \\
 &= (c-1) \sum_{i\in \mathbb{Z}} |\{u\in \OPT \colon \ell(u) \geq i\}| \\
 &= (c-1) \left[\sum_{i=-\infty}^{\overline{b}-1} c^i \cdot |\{u\in \OPT\colon \ell(u) \geq i\}| + \sum_{i=\overline{b}}^q c^i \cdot |\{u\in \OPT \colon \ell(u) \geq i\}|\right]\enspace.
\end{aligned}
\end{equation}
To upper bound the terms in the first sum, we use
\begin{equation}\label{eq:boundOPTEllRank}
|\{u\in \OPT\colon \ell(u) \geq i\}| \leq \rank(M) \qquad \forall i\in \mathbb{Z} \enspace,
\end{equation}
which holds because $\{u\in \OPT\colon \ell(u) \geq i\}\subseteq \OPT$ and $\OPT\in \cI$.
Moreover, for the second sum, we use
\begin{equation}\label{eq:boundOptEllAi}
|\{u\in \OPT\colon \ell(u) \geq i\}|\leq |\overline{A}_i| \qquad \forall i\in \mathbb{Z}_{\geq \overline{b}} \enspace,
\end{equation}
which holds due to the following.
By the definition of $\ell(u)$ and \cref{lem:Ai_spanned_by_previous}, we have $\{u\in \OPT: \ell(u) \geq i\}\subseteq \Span(\overline{A}_i)$.
\cref{eq:boundOptEllAi} now follows from 
\begin{equation*}
|\{u\in \OPT\colon \ell(u) \geq i\}| = \rank(\{u\in \OPT \colon \ell(u) \geq i\}) \leq \rank(\Span(\overline{A}_i)) = \rank(\overline{A}_i) = |\overline{A}_i|\enspace,
\end{equation*}
where the first equality holds because $\{u\in \OPT\colon \ell(u) \geq i\} \subseteq \OPT \in \cI$, the inequality holds because $\{u\in \OPT\colon \ell(u) \geq i\} \subseteq \Span(\overline{A}_i)$, and the last equation follows from $\overline{A}_i\in \cI$.

We now combine the above-proved inequalities to obtain the desired result:
\begin{align*}
\sum_{u\in \OPT} c^{\ell(u) + 1}
  &\leq  (c-1) \left[\sum_{i=-\infty}^{\overline{b}-1} c^i\cdot \rank(M)  +  \sum_{i=\overline{b}}^q c^i \cdot |\overline{A}_i|\right] \\
  &= (c-1) \left[\rank(M) \cdot \frac{c^{\overline{b}}}{c-1} + \sum_{i=\overline{b}}^{q} c^i \cdot |\overline{A}_i| \right] \\
  &\leq (c-1) \left(1+\frac{\eps}{2c}\right) \sum_{i=\overline{b}}^q c^i \cdot |\overline{A}_i| \\
  &\leq (c-1) \left(1+\frac{\eps}{2c}\right) \frac{m}{1-c^{-m}} F(s)\enspace, \\
\end{align*}
where the first inequality follows by \cref{eq:rewriteSumCOpt}, \cref{eq:boundOPTEllRank}, and \cref{eq:boundOptEllAi},
the second inequality follows by \cref{lem:lowLevelsAreNegligible},
and the last one is a consequence of \cref{lem:lowerBoundFsInFAall}.
\end{proof}

Combining \cref{lem:diffOptToAall,lem:upperBoundSumCOpt} we obtain the following lower bound on $F(s)$ in terms of $f(\OPT)$.
\begin{corollary}\label{cor:boundFsInOpt}
\begin{equation*}
F(s) \geq \frac{1-c^{-m}-\frac{\eps}{2c}}{(c-1)\left(1+\frac{\eps}{2c}\right)m + 1} \cdot f(\OPT)\enspace.
\end{equation*}
\end{corollary}
\begin{proof}
\begin{align*}
f(\OPT) &\leq (c-1)\left(1+\frac{\eps}{2c}\right)\frac{m}{1-c^{-m}} \cdot F(s) + F(\aall) \\
        &\leq (c-1)\left(1+\frac{\eps}{2c}\right)\frac{m}{1-c^{-m}} \cdot F(s) + \left(1-c^{-m} - \frac{\eps}{2c}\right)^{-1} F(s)\\
        &\leq \frac{(c-1)\left(1+\frac{\eps}{2c}\right)m + 1}{1-c^{-m}-\frac{\eps}{2c}} \cdot F(s)\enspace, \\
\end{align*}
where the first inequality follows from \cref{lem:diffOptToAall,lem:upperBoundSumCOpt},
and the second one from \cref{lem:lowerBoundFsInFAall}.
\end{proof}

Our result for {\MSMM}, i.e., \cref{thm:mainSingleMatMon}, now follows from \cref{cor:boundFsInOpt} and our choice of parameters $\alpha$, $m$, and $c$, which have been chosen to optimize the ratio. This leads to a lower bound on $F(s)$ in terms of $f(\OPT)$, which in turn leads to a lower bound on $f(R)$ because $f(R) \geq F(s)$.
\begin{proof}[Proof of \cref{thm:mainSingleMatMon}]
We have
\begin{equation}\label{eq:boundFsViaOptStep1}
\begin{aligned}
f(R)		 &\geq F(s)  \\
         &\geq \frac{1-c^{-m}-\frac{\eps}{2c}}{(c-1)\left(1+\frac{\eps}{2c}\right)m + 1} f(\OPT)\\
         &\geq \frac{1 - e^{-\alpha} - \frac{\eps}{2c}}{(c-1)\cdot \left(1+\frac{\eps}{2c}\right)m+1} \cdot f(\OPT)      \\
         &= \frac{1 - e^{-\alpha} - \frac{\eps}{2c}}{\alpha \left(c + \frac{\eps}{2}\right) + 1} \cdot f(\OPT)                 \\
         &\geq \left(\frac{1 - e^{-\alpha}}{\alpha + 1} \cdot \frac{1}{c+\frac{\eps}{2}} - \frac{\eps}{2c}\right)\cdot f(\OPT) \\
         &\geq \left(\frac{1}{\alpha+2} \cdot \frac{1}{c+\frac{\eps}{2}} - \frac{\eps}{2c}\right)\cdot f(\OPT)\enspace,
\end{aligned}
\end{equation}
where the second inequality is due to \cref{cor:boundFsInOpt},
the third one follows from $c^{-m} = (1-\sfrac{\alpha}{m})^m \leq e^{-\alpha}$,
the equality uses again $c = \frac{m}{m-\alpha}$, and the last inequality uses our choice of value for $\alpha$ (note the inequality would have held as an equality if $\alpha$ was obeying $e^{\alpha} = \alpha + 2$, and we chose a value that is close).

By our choice of $m=\lceil \sfrac{3\alpha}{\eps} \rceil$, we obtain the following bound on $c$:
\begin{equation*}
c = \frac{m}{m-\alpha} = \frac{1}{1-\frac{\alpha}{m}} \leq \frac{1}{1-\frac{\eps}{3}} \leq 1 + \frac{\eps}{2}\enspace.
\end{equation*}
Plugging this bound into \cref{eq:boundFsViaOptStep1} and using $c\geq 1$, we get
\begin{align*}
F(R) &\geq \left(\frac{1}{\alpha + 2} \cdot \frac{1}{1+\eps} - \frac{\eps}{2} \right)\cdot f(\OPT) \\
     &\geq \left(\frac{1}{\alpha + 2}\cdot (1-\eps) - \frac{\eps}{2} \right)\cdot f(\OPT)          \\
     &\geq \left(\frac{1}{\alpha + 2} - \eps\right)\cdot f(\OPT)\enspace,
\end{align*}
as desired.
\end{proof}

\section{Framework for Multi-pass Algorithms}
\label{sec:framework}
In this section we present the details of the framework used to prove our $(1 - 1/e)$-approximation results (\cref{thm:multipass,thm:multipass_random}). We remind the reader that the proofs of these theorems (using the framework) can be found in \cref{sec:adversarial_local_search,sec:random_local_search}, respectively.
Badanidiyuru and Vondr\'{a}k~\cite{badanidiyuru2014fast} described an algorithm called ``Accelerated Continuous Greedy'' that obtains an approximation guarantee of $1 - 1/e - O(\eps)$ for {\MSMM} for every $\eps \in (0, 1)$. Their algorithm is not a data stream algorithm, but it enjoys the following nice properties.
\begin{itemize}
	\item The algorithm includes a procedure called ``Decreasing-Threshold Procedure''. This procedure is the only part of the algorithm that directly accesses the input.
	\item The Decreasing-Threshold Procedure is called $O(\eps^{-1})$ times during the execution of the algorithm.
	\item In addition to the space used by this procedure, Accelerated Continuous Greedy uses only space that is linear in the space necessary to store the outputs of the various executions of the Decreasing-Threshold Procedure.
	\item The Decreasing-Threshold Procedure returns a base $D$ of $M$ after every execution, and this base is guaranteed to obey \cref{eq:procedure_guarantee} stated below. The analysis of the approximation ratio of Accelerated Continuous Greedy treats Decreasing-Threshold Procedure as a black box except for the fact that its output is a base $D$ of $M$ obeying \cref{eq:procedure_guarantee}, and therefore, this analysis will remain valid even if Decreasing-Threshold Procedure is replaced by any other algorithm with the same guarantee. Furthermore, one can verify that the analysis continues to work (with only minor technical changes) even if the output $D$ of the replacing algorithm obeys \cref{eq:procedure_guarantee} only in expectation.
\end{itemize}

Let us now formally state the property that the output base of Decreasing-Threshold Procedure obeys. Let $P_M$ be the matroid polytope of $M$, and let $F$ be the multilinear extension of $f$. Decreasing-Threshold Procedure gets as input a point $x \in (1 - \eps) \cdot P_M$, and its output base $D$ is guaranteed to obey
\begin{equation} \label{eq:procedure_guarantee}
	F(x') - F(x)
	\geq
	\eps[(1 - 3\eps) \cdot f(\OPT) - F(x')]
	\enspace,
\end{equation}
where $x' = x + \eps \cdot \characteristic_D$ and $\OPT$ denotes an optimal solution. 

Our objective in \cref{sec:adversarial_local_search,sec:random_local_search} is to describe semi-streaming algorithms that can function as replacements for the offline procedure Decreasing-Threshold Procedure. The next proposition shows that plugging such a replacement into Accelerated Continuous Greedy yields a roughly $(1 - 1/e)$-approximation semi-streaming algorithm.
\begin{proposition} \label{prop:equation_enough}
Assume there exists a semi-streaming algorithm that given a point $x \in (1 - \eps) \cdot P_M$ makes $p$ passes over the input stream, stores $O(k / \eps)$ elements, and outputs a base $D$ obeying \cref{eq:procedure_guarantee} in expectation. Then, there exists a semi-streaming algorithm for maximizing a non-negative monotone submodular function subject to a matroid constraint of rank $k$ that stores $O(k/\eps)$ elements, makes $O(p/\eps)$ many passes and achieves an approximation guarantee of $1 - 1/e - \eps$.
\end{proposition}
\begin{proof}
Observe that the proposition is trivial when $\eps \geq 1 - 1/e$, and therefore, we assume below that $\eps < 1 - 1/e$. Furthermore, for simplicity, we describe an algorithm with an approximation ratio of $1 - 1/e - O(\eps)$ rather than a clean ratio of $1 - 1/e - \eps$. However, one can switch between the two ratios by scaling $\eps$ by an appropriate constant.

Let us denote by $\ALG$ the algorithm whose existence is promised by the statement of the proposition, and consider an algorithm called \texttt{Data Stream Continuous Greedy} (\DSCG) obtained from the Accelerated Continuous Greedy algorithm of~\cite{badanidiyuru2014fast} when every execution of the Decreasing-Threshold Procedure by the last algorithm is replaced with an execution of {\ALG}. We explain below why {\DSCG} has all the properties guaranteed by the proposition. We begin by recalling that since the approximation ratio analysis of Accelerated Continuous Greedy in~\cite{badanidiyuru2014fast} treats the Decreasing-Threshold Procedure as a black box that in expectation has the guarantee stated in \cref{eq:procedure_guarantee}, and {\ALG} also has this guarantee, this analysis can be applied as is also to {\DSCG}, and therefore, {\DSCG} is a $(1 - 1/e - O(\eps))$-approximation algorithm.

Recall now that Accelerated Continuous Greedy accesses its input only through the Decreasing-Threshold Procedure, which implies that {\DSCG} is a data stream algorithm just like {\ALG}. Furthermore, since Accelerated Continuous Greedy accesses the Decreasing-Threshold Procedure $O(\varepsilon^{-1})$ times, the number of passes used by {\DSCG} is larger by a factor of $O(\eps^{-1})$ compared to the number of passes used by {\ALG} (which is denoted by $p$). Hence, {\DSCG} uses $O(p/\eps)$ passes.

It remains to analyze the space complexity of {\DSCG}. Since Accelerated Continuous Greedy uses space of size linear in the space necessary to keep the $O(\eps^{-1})$ bases that it receives from the Decreasing-Threshold Procedure, the space complexity of {\DSCG} is larger than the space complexity of the semi-streaming algorithm {\ALG} only by an additive term of $\widetilde{O}(k/\eps)$. As this term is nearly linear in $k$ for any constant $\eps$, we get that {\DSCG} has a low enough space complexity to be a semi-streaming algorithm. Furthermore, since the $O(\eps^{-1})$ bases that {\DSCG} gets from $ALG$ can include only $O(k/\eps)$ elements, this expression upper bounds the number of elements stored by {\DSCG} in addition to the $O(k/\eps)$ elements stored by {\ALG} itself.
\end{proof}

It turns out that one natural way to get a base $D$ obeying \cref{eq:procedure_guarantee} is to output a local maximum with respect to the objective function $g(S) = F(x + \eps \cdot \characteristic_S)$ (i.e., a base $D$ whose value with respect to this objective cannot be improved by replacing an element of $D$ with an element of $\cN \setminus D$). Getting such a maximum using a semi-streaming algorithm with a reasonable number of passes is challenging; however, one can define weaker properties that still allow us to get \cref{eq:procedure_guarantee}. Specifically, for any $\eps \in (0, 1)$, we say that a set $D$ is an \emph{$\eps$-approximate local maximum} with respect to $g$ if
\[
    g(D \mid \varnothing) \geq g(B \mid D) + \sum_{u \in B \cap D} \mspace{-9mu} g(u \mid D - u) - \eps \cdot g(\OPT_g \mid \varnothing)
\]
for every base $B$ of $M$, where $\OPT_g$ is a base maximizing $g$. (Intuitively, one should think of $B$ as being the optimal solution with respect to $f$.)

One property of an approximate local maximum is that its value (with respect to $g$) is an approximation to $g(\OPT_g)$.
\begin{observation} \label{obs:local_maximum_approximation}
For every $\eps \in (0, 1)$, if $D$ is an $\eps$-approximate local maximum with respect to $g$, then $g(D) \geq \frac{1 - \eps}{2} \cdot g(OPT_g)$.
\end{observation}
\begin{proof}
One can verify that the non-negativity, monotonicity and submodularity of $f$ implies that $g$ also has these properties. Therefore,
\begin{align*}
    g(D)
    \geq{} &
    g(D \mid \varnothing)
    \geq
    g(OPT_g \mid D) + \sum_{u \in \OPT_g \cap D} \mspace{-18mu} g(u \mid D - u) - \eps \cdot g(\OPT_g \mid \varnothing)\\
    \geq{} &
    g(\OPT_g \mid D) - \eps \cdot g(OPT_g \mid \varnothing)
    \geq
    (1 - \eps) \cdot g(\OPT_g) - g(D)
    \enspace,
\end{align*}
where the first inequality holds by the non-negativity of $g$, the second inequality follows from the fact that $D$ is an $\eps$-approximate local maximum (for $B = \OPT_g$), the third inequalities follow from the monotonicity of $g$, and the last inequality hold by $g$'s non-negativity and monotonicity. Rearranging the above inequality now yields the observation.
\end{proof}
Using the last observation we can prove that any approximate local maximum with respect to $g$ obeys \cref{eq:procedure_guarantee}, and the same holds also for any solution that is almost as good as some approximate local maximum.
\begin{lemma} \label{lem:local_maximum_good}
For every $\eps \in (0, 1)$, if $D'$ is an $\eps$-approximate local maximum with respect to $g$, then any (possibly randomized) set $D$ such that $\bE[g(D \mid \varnothing)] \geq (1 - \eps) \cdot g(D' \mid \varnothing)$ obeys \cref{eq:procedure_guarantee} in expectation. In particular, this is the case for $D = D'$ since the monotonicity of $f$ implies that $g$ is non-negative.
\end{lemma}
\begin{proof}
We need to consider two cases. The simpler case is when $g(\OPT_g \mid \varnothing) \geq 2 \eps \cdot f(\OPT)$, where we recall that $OPT$ is an optimal base with respect to $f$. Since $x' = x + \eps \cdot \characteristic_D$ by definition, we get in this case
\begin{align*}
	\bE[F(x')] &{}- F(x)
	=
	\bE[F(x + \eps \cdot \characteristic_D)] - F(x)
	=
	\bE[g(D \mid \varnothing)]
	\geq
	(1 - \eps) \cdot g(D' \mid \varnothing)\\
	\geq{} &
	\tfrac{(1 - \eps)^2}{2}g(\OPT_g \mid \varnothing)
	\geq
	\eps(1 - 2\eps) \cdot f(\OPT)
	\geq
	\eps((1 - 3\eps) \cdot f(\OPT) - \bE[F(x')])
	\enspace,
\end{align*}
where the second inequality holds by Observation~\ref{obs:local_maximum_approximation}.

In the rest of the proof we consider the case of $g(\OPT_g \mid \varnothing) \leq 2 \eps \cdot f(\OPT)$. We note that, in this case,
\begin{align*}
	\frac{\bE[F(x')] - F(x)}{1 - \eps}
	={} &
	\frac{\bE[F(x + \eps \cdot \characteristic_D)] - F(x)}{1 - \eps}
	=
	\frac{\bE[g(D \mid \varnothing)]}{1 - \eps}
	\geq
	g(D' \mid \varnothing)\\
	\geq{} &
	g(\OPT \mid D') + \sum_{u \in \OPT \cap D'} \mspace{-18mu} g(u \mid D' - u) - \eps \cdot g(\OPT_g \mid \varnothing)\\
	\geq{} &
	g(\OPT \mid D') + \sum_{u \in \OPT \cap D'} \mspace{-18mu} g(u \mid D' - u) - 2\eps^2 \cdot f(\OPT)
	\enspace,
\end{align*}
where the second inequality holds since $D'$ is an $\eps$-approximate local maximum (by plugging $B = \OPT$ into the definition of such maxima). Let us now further develop the first two terms on the rightmost side of the last inequality. By the submodularity and monotonicity of $f$, if we denote $y = x + \eps \cdot \characteristic_{D'}$, then
\begin{align*}
    g(\OPT\mspace{-9mu}&\mspace{9mu}{} \mid D') + \sum_{u \in \OPT \cap D'} \mspace{-18mu} g(u \mid D' - u)\\
    ={} &
    F(x + \eps \cdot \characteristic_{\OPT \cup D'}) - F(x + \eps \cdot \characteristic_{D'})
    + 
    \sum_{u \in \OPT \cap D'} \mspace{-9mu} [F(x + \eps \cdot \characteristic_{D'}) - F(x + \eps \cdot \characteristic_{D' - u})]\\
    \geq{} &
    F(y + \eps \cdot \characteristic_{\OPT \setminus D'}) - F(y)
    +
    \sum_{u \in \OPT \cap D'} \mspace{-9mu} [F((y + \eps \cdot \characteristic_{\{u\}}) \wedge \characteristic_\cN) - F(y)]\\
    \geq{} &
    F((y + \eps \cdot \characteristic_{\OPT}) \wedge \characteristic_\cN) - F(y)
    \enspace.
\end{align*}
Combining the last two inequalities yields
\begin{align*}
	\bE[F(x')] - F(x)
	\geq{} &
	(1- \eps) [F((y + \eps \cdot \characteristic_{\OPT}) \wedge \characteristic_\cN) - F(y)] - 2\eps^2 \cdot f(\OPT)\\
	\geq{} &
	(1- \eps) [F(y + \eps((\characteristic_\cN - y) \wedge \characteristic_{\OPT})) - F(y)] - 2\eps^2 \cdot f(\OPT)\\
	\geq{} &
	\eps(1- \eps)  [F(y \vee \characteristic_{\OPT}) - F(y)] - 2\eps^2 \cdot f(\OPT)\\
	\geq{} &
	\eps ((1- \eps) f(\OPT) - \bE[F(x')]) - 2\eps^2 \cdot f(\OPT)\\
	={} &
	\eps \cdot ((1 - 3\eps) f(\OPT) - \bE[F(x')])
	\enspace,
\end{align*}
where the second inequality holds by the monotonicity of $f$, the third inequality holds because the submodularity of $f$ guarantees that $F$ is concave along non-negative directions (such as $(\characteristic_\cN - y) \wedge \characteristic_{\OPT}$) and the last inequality holds by the motonicity of $f$ and the observation that
\[
    F(y)
    =
    g(D')
    =
    g(\varnothing) + g(D' \mid \varnothing)
    \leq
    g(\varnothing) + \frac{\bE[g(D \mid \varnothing)]}{1 - \eps}
    \leq
    \frac{\bE[g(D)]}{1 - \eps}
    =
    \frac{\bE[F(x')]}{1 - \eps}
    \enspace.
    \qedhere
\]
\end{proof}

In \cref{sec:adversarial_local_search} we describe a semi-streaming algorithm that can be used to find an $\eps$-approximate local maximum of a non-negative monotone submodular function. By applying this algorithms to $g$, we get (via Lemma~\ref{lem:local_maximum_good}) an algorithm having all the properties assumed by Proposition~\ref{prop:equation_enough}; which proves \cref{thm:multipass}. In \cref{sec:random_local_search} we attempt to use the same approach to get a result for random order streams. However, in this setting we are not able to guarantee an $\eps$-approximate local maximum. Instead, we design an algorithm whose output has in expectation a value that is almost as good as the value of the worst approximate local maximum. This leads to a proof of \cref{thm:multipass_random}.

\section{Approximate Local Maximum for Adversarial Streams} \label{sec:adversarial_local_search}

In this section we prove following proposition, which guarantees the existence of a semi-streaming multi-pass algorithm for finding an $\eps$-approximate local maximum in adversarial streams, i.e., when the order of the elements in the input stream is arbitrary. We note that this section highly depends on \cref{sec:framework}, and should not be read before that section.
\begin{restatable}{proposition}{propLocalMaxAdversarial} \label{prop:local_max_adversarial}
For every constant $\eps > 0$, there is a multi-pass semi-streaming algorithm that given an instance of {\MSMM}  with a matroid of rank $k$ stores $O(k)$ elements, makes $O(\eps^{-2})$ many passes, and outputs an $\eps$-approximate local maximum.
\end{restatable}

By \cref{lem:local_maximum_good} and \cref{prop:equation_enough}, the last proposition implies \cref{thm:multipass}. Therefore, we concentrate in this section on proving \cref{prop:local_max_adversarial}. Specifically, in \cref{ssc:local_search_pass} we reanalyze a known single pass algorithm for {\MSMM}, and then we use the results of this reanalysis to show in \cref{ssc:local_search_multi} that by executing this known algorithm multiple times one can obtain the algorithm whose existence is guaranteed by Proposition~\ref{prop:local_max_adversarial}. We conclude with \cref{ssc:lower_bound_multipass}, which provides a detailed discussion of the result of McGregor and Vu~\cite{mcgregor2019better} mentioned in \cref{sec:introduction}. We recall that this result shows that \cref{thm:multipass} is tight in the sense that any data stream algorithm that obtains a constant approximation ratio $\alpha > 1 - 1/e$ for {\MSMM} and uses a constant number of passes must have a space complexity that is linear in $|\cN|$.

\subsection{Alternative Analysis for a Known Single Pass Algorithm} \label{ssc:local_search_pass}

The first data stream algorithm for {\MSMM} was described by Chakrabarti and Kale~\cite{chakrabarti2015submodular}. In this section we consider a variant of their algorithm. This variant is a special case of an algorithm that was described by Huang, Thiery and Ward~\cite{HuangTW20} (based on ideas of Chekuri et al.~\cite{chekuri2015streaming}), and it appears as \cref{alg:local_search_pass}. \cref{alg:local_search_pass} gets a parameter $c > 1$ and a base $S_0$ of $M$ that it starts from, and intuitively, it inserts every arriving element into its solution (at the expense of an appropriate existing element) whenever such a swap is beneficial enough in some sense. In the pseudocode of \cref{alg:local_search_pass}, we denote by $u_1, u_2, \dotsc, u_n$ the elements of $\cN \setminus S_0$ in the order of their arrival. Similarly, we denote by $S_i$ the solution of the algorithm immediately after it processes element $u_i$ for every integer $1 \leq i \leq n$. Finally, we denote the elements of the base $S_0$ by $u_{1 - |S_0|}, u_{2 - |S_0|}, \dotsc, u_0$ in an arbitrary order. This notation allows us to define, for every integer $1 - |S_0| \leq i \leq n$ and set $T \subseteq \cN$,
\[
	f(u_i : T)
	=
	f(u_i \mid \{u_j \in T \mid 1 - |S_0| \leq j < i\})
	\enspace.
\]
In other words, $f(u_i : T)$ is the marginal contribution of $u_i$ with respect to the elements of $T$ that appear in the input stream of the algorithm before $u_i$.

\begin{algorithm} 
\caption{\textsc{Single Local Search Pass} $(S_0, c)$} \label{alg:local_search_pass}
\begin{algorithmic}[1]
\For{every element $u_i \in \cN \setminus S_0$ that arrives}
	\State Let $C_i$ be the single cycle in $S_{i - 1} + u_i$.
	\State Let $u'_i$ be the element in $C_i - u_i$ minimizing $f(u'_i : S_{i - 1})$. \\\hspace{4cm}\Comment{Note that $u'_i$ is equal to $u_j$ for some $j < i$.}
	\If{$f(u_i \mid S_{i - 1}) \geq c \cdot f(u'_i : S_{i - 1})$} 
		\State Set $S_i \gets S_{i - 1} - u'_i + u_i$.
	\Else{}	
		\State Set $S_i \gets S_{i - 1}$.
	\EndIf
\EndFor
\State \Return{$S_n$}.
\end{algorithmic}
\end{algorithm}

One can verify that the solution of \cref{alg:local_search_pass} remains a base of $M$ throughout the execution of the algorithm. Furthermore, it is known that \cref{alg:local_search_pass} achieves $4$-approximation for {\MSMM} when $c = 2$. However, we need to prove a slightly different property of it. Specifically, we show below that when $S_0$ is not an approximation local maximum, the value of the final solution $S_n$ of \cref{alg:local_search_pass} is much larger than the value of the initial solution $S_0$. In \cref{ssc:local_search_multi} we show how this property of \cref{alg:local_search_pass} can be used to find an $\eps$-approximate local maximum in $O(\eps^{-2})$ passes.

Let $B$ be an arbitrary base of $M$ (intuitively, one can think of $B$ as the optimal solution, although this will not always be the case). We begin the analysis of \cref{alg:local_search_pass} by showing a lower bound on the sum of the marginal contributions of the elements of $B \setminus S_0$ with respect to the solutions held by \cref{alg:local_search_pass} when these elements arrive. Let $A$ be the set of all elements that belong to the solution of \cref{alg:local_search_pass} at some point (formally, $A = \bigcup_{i = 0}^n S_i$).

\begin{lemma} \label{lem:adding_bound}
$\sum_{u_i \in B \setminus S_0} f(u_i \mid S_{i - 1}) \geq f(S_0 \cup B) + \tfrac{1}{c-1} \cdot f(S_0) - \tfrac{c}{c-1} \cdot f(S_n)$.
\end{lemma}
\begin{proof}
By the submodularity of $f$,
\begin{align*}
	\sum_{u_i \in B \setminus S_0} \mspace{-11mu} f(u_i \mid {}&S_{i - 1})
	\geq
	\sum_{u_i \in B \setminus S_0} \mspace{-11mu} f(u_i \mid A)
	\geq
	f(B \mid A)
	=
	f(B \cup (A \setminus S_0) \mid S_0) - f(A \setminus S_0 \mid S_0)\\
	\geq{}&
	f(B \mid S_0) - \sum_{u_i \in A \setminus S_0} \mspace{-11mu} f(u_i \mid S_0 \cup \{u_j \in A \mid 1 \leq j < i\})\\
	\geq{} &
	f(B \mid S_0) - \sum_{u_i \in A \setminus S_0} \mspace{-11mu} f(u_i \mid S_{i - 1})
	\enspace,
\end{align*}
where the third inequality holds by the monotonicity of $f$.

Let us now upper bound the second term in the rightmost side. Since all the elements of $A \setminus S_0$ were accepted by \cref{alg:local_search_pass} into its solution,
\begin{align*}
	\sum_{u_i \in A \setminus S_0} \mspace{-9mu} f(u_i \mid S_{i - 1})
	\leq{} &
	\tfrac{c}{c-1} \cdot \sum_{u_i \in A \setminus S_0} \mspace{-9mu} [f(u_i \mid S_{i - 1}) - f(u'_i : S_{i - 1})]\\
	\leq{} &
	\tfrac{c}{c-1} \cdot \sum_{u_i \in A \setminus S_0} \mspace{-9mu} [f(u_i \mid S_{i - 1}) - f(u'_i \mid S_{i - 1} + u_i - u'_i)]\\
	={} &
	\tfrac{c}{c-1} \cdot \sum_{u_i \in A \setminus S_0} \mspace{-9mu} [f(S_i) - f(S_{i - 1})]
	=
	\tfrac{c}{c-1} \cdot [f(S_n) - f(S_0)]
	\enspace,
\end{align*}
where the second inequality holds by the submodularity of $f$, and the last equality holds since $S_i = S_{i - 1}$ for every integer $1 \leq i \leq n$ for which $u_i \not \in A$. The lemma now follows by combining the two above inequalities.
\end{proof}

To complement the last lemma, we need to upper bound the marginal contributions of the elements $u'_i$ corresponding to the elements $u_i \in A$ with respect to the solutions of \cref{alg:local_search_pass} when the last elements arrive. We prove such an upper bound in \cref{cor:removing_bound} below. However, proving this upper bound requires us to present a few additional definitions as well as properties of the objects defined. We begin by constructing an auxilary directed graph $G$ whose vertices are the elements of $\cN$. Furthermore, for every element $u_i \in \cN \setminus S_0$, we create edges for the graph $G$ in the following way. Note that there is a single element $c_i \in C_i$ that does not belong to $S_i$. The graph $G$ includes edges from $c_i$ to every other element of $C_i$. Let us now prove some properties of the graph $G$.

\begin{observation} \label{obs:edge_increase}
For every element $u \in \cN$, let us define
\[
	\val(u)
	=
	\begin{cases}
		f(u : S_n) & \text{if $u \in S_n$} \enspace, \\
		f(u : S_i) & \begin{aligned} & \text{if $u \in A \setminus S_n$, and $u$ was removed from} \\ & \text{the solution of \cref{alg:local_search_pass} when $u_i$ arrived} \enspace,\end{aligned} \\
		f(u \mid S_i) & \text{if $u \not \in A$ and $u = u_i$} \enspace.
	\end{cases}
\]
Then, for every edge $uv$ of $G$ such that $u \in A$, $\val(u) \leq \val(v)$.
\end{observation}
\begin{proof}
Since there is an edge from $u$ to $v$, $u$ must have been removed from $A$ when some element $u_i$ arrived, and $v$ was another element of the cycle $C_i$. If $v \neq u_i$, then the fact that $u$ was removed (rather than $v$) implies
\[
	\val(u)
	=
	f(u : S_i)
	\leq
	f(v : S_i)
	\leq
	\val(v)
	\enspace,
\]
where the second inequality holds since $\val(v)$ is equal to $f(v : S_j)$ for some $j \geq i$. Otherwise, if $v = u_i$, then the fact that $u$ was removed following the arrival of $v$ implies
\[
	\val(u)
	=
	f(u : S_i)
	\leq
	\frac{f(v \mid S_i)}{c}
	\leq
	\frac{\val(v)}{c}
	\leq
	\val(v)
	\enspace,
\]
where the last inequality holds since the monotonicity of $f$ guarantees that $\val(v)$ is non-negative.
\end{proof}

\begin{corollary} \label{cor:path_increase}
If $u$ and $v$ are two elements of $A$ such that $v$ is reachable from $u$ in $G$, then $\val(u) \leq \val(v)$.
\end{corollary}
\begin{proof}
The corollary follows from \cref{obs:edge_increase} because the construction of $G$ guarantees that the vertices of $\cN \setminus A$ are all sources of $G$ (i.e., vertices that do not have any edge entering them).
\end{proof}

\begin{observation}
$G$ is acyclic; and every element $u \in \cN$ that is not a sink of $G$ is spanned by the elements of $\delta^+(u)$, where $\delta^+(u) = \{v \mid \text{$uv$ is an edge of $G$}\}$.
\end{observation}
\begin{proof}
Every edge $e$ of $G$ was created due to some cycle $C_i$. Furthermore, the edge $e$ goes from a vertex that does not appear in $S_i$ or any solution that \cref{alg:local_search_pass} has at a later time point to a vertex that does belong to $S_i$. Therefore, if we sort the vertices of $G$ by the largest index $i$ for which they belong to $S_i$ (a vertex that does not belong to $S_i$ for any $i$ is placed before all the vertices that do belong to $S_i$ for some $i$), then we obtain a topological order of $G$, which implies that $G$ is acyclic.

To prove the second part of the observation, we note that whenever the construction of $G$ includes edges leaving a node $u$, this implies that these edges go to all the vertices of $C - u$ for some cycle $C$ that includes $u$. Therefore, $\delta^+(u) \supseteq C - u$ spans $u$.
\end{proof}

To use the last observation, we need the following known lemma (a similar lemma appeared earlier in~\cite{badanidiyuru2011buyback} in an implicit form, and was made explicit in~\cite{chekuri2015streaming}).
\begin{lemma}[Lemma 13 of~\cite{feldman2018do}] \label{lem:mapping}
Consider an  arbitrary  directed  acyclic  graph $G = (V,E)$ whose  vertices  are elements of some matroid $M'$.  If every non-sink vertex $u$ of $G$ is spanned by $\delta^+(u)$ in $M'$, then for every set $S$ of vertices of $G$ which is independent in $M'$ there must exist an injective function $\psi_S$ such that, for every vertex $u \in S$, $\psi_S(u)$ is a sink of $G$ which is reachable from $u$.
\end{lemma}

\begin{corollary} \label{cor:removing_bound}
$\sum_{u_i \in B \setminus A} f(u'_i : S_{i - 1}) + c \cdot \sum_{u_i \in B \cap (A \setminus S_0)} f(u'_i : S_{i - 1}) \leq f(S_n \mid \varnothing) - f(S_0 \mid S_0 \setminus B)$.
\end{corollary}
\begin{proof}
Let $\psi_B$ be the function whose existence is guaranteed by \cref{lem:mapping} (recall that $B$ is a base of $M$, and therefore, is independent in $M$). Consider now an element $u_i \in B \setminus A$, and let $P_i$ be the path in $G$ from $u_i$ to $\psi_B(u_i)$ whose existence is guaranteed by \cref{lem:mapping}. If we denote by $u''_i$ the element that appears in this path immediately after $u_i$ (there must be such an element because $u_i \not \in A \supseteq S_n$, and therefore, is not a sink of $G$), then $\val(u''_i) \leq \val(\psi_B(u_i))$ according to \cref{cor:path_increase}. Additionally, since $u_i$ was rejected by \cref{alg:local_search_pass} immediately upon arrival, both $u'_i$ and $u''_i$ are elements of $C_i - u_i$, and thus, due to the way in which \cref{alg:local_search_pass} selects $u'_i$,
\[
	f(u'_i : S_{i - 1})
	\leq
	f(u''_i : S_{i - 1})
	\leq
	\val(u''_i)
	\leq
	\val(\psi_B(u_i))
	=
	f(\psi_B(u_i) : S_n)
	\enspace,
\]
where the last equality holds since $\psi_B(u_i)$ is a sink of $G$, and therefore, belongs to $S_n$.

Consider now an element $u_i \in B \cap (A \setminus S_0)$. Since $\psi_B(u_i)$ is reachable in $G$ from $u_i$, $\val(u_i) \leq \val(\psi_B(u_i))$. Therefore, the fact that $u_i$ was added upon arrival to the solution of \cref{alg:local_search_pass} implies
\[
	f(u'_i : S_{i - 1})
	\leq
	\frac{f(u_i \mid S_{i - 1})}{c}
	\leq
	\frac{\val(u_i)}{c}
	\leq
	\frac{\val(\psi_B(u_i))}{c}
	=
	\frac{f(\psi_B(u_i) : S_n)}{c}
	\enspace.
\]

Combining both the above inequalities, we get
\begin{align*}
	\sum_{u_i \in B \setminus A} f(u'_i&{} : S_{i - 1}) + c \cdot \sum_{u_i \in B \cap (A \setminus S_0)} \mspace{-18mu} f(u'_i : S_{i - 1})
	\leq
	\sum_{u_i \in B \setminus S_0} \mspace{-9mu} f(\psi_B(u_i) : S_n)\\
	={} &
	\sum_{u_i \in S_n} f(\psi_B(u_i) : S_n) - \sum_{u_i \in B \cap S_0} \mspace{-9mu} f(\psi_B(u_i) : S_n)
	\leq
	f(S_n \mid \varnothing) - \sum_{u_i \in B \cap S_0} \mspace{-9mu} \val(u_i)\\
	\leq{} &
	f(S_n \mid \varnothing) - f(S_0 \cap B \mid S_0 \setminus B)
	=
	f(S_n \mid \varnothing) - f(S_0 \mid S_0 \setminus B)
	\enspace,
\end{align*}
where the first equality holds because $\psi_B$ is a bijection from $B$ to $S_n$ by \cref{lem:mapping}, and the third inequality holds since $\sum_{u_i \in B \cap S_0} \val(u_i) \geq \sum_{u_i \in B \cap S_0} f(u_i : S_0) \geq \sum_{u_i \in B \cap S_0} f(u_i \mid (S_0 \setminus B) \cup (S_0 \cap \{u_{1 - k}, u_{2 - k}, \dotsc, u_{i - 1}\})) = f(S_0 \cap B \mid S_0 \setminus B)$.
\end{proof}

We are now ready to state and prove our main result regarding \cref{alg:local_search_pass}. In the beginning of the section, we claimed that this result shows that the difference $f(S_n) - f(S_0)$ is large whenever $S_0$ is not an approximate local maximum. To see why this is the case, note that the rightmost side in the next proposition is guaranteed to be large (for some base $B$) when $S_0$ is not an approximate local maximum.
\begin{proposition} \label{prop:single_pass}
\cref{alg:local_search_pass} is a semi-streaming algorithm, and it outputs a base $S_n$ of $M$ that obeys $(c - 1) \cdot f(S_n \mid \varnothing) + \tfrac{3c - 2}{c - 1}[f(S_n) - f(S_0)] \geq	f(B \mid S_0 \setminus B) - f(S_0 \mid \varnothing) \geq f(B \mid S_0) + \sum_{u \in B \cap S_0} f(u \mid S_0 - u) - f(S_0 \mid \varnothing)$.
\end{proposition}
\begin{proof}
The first part of the proposition holds because implementing \cref{alg:local_search_pass} requires us to maintain only two bases of the matroid $M$, the input base $S_0$ and the current solution of the algorithm. The rest of this proof is devoted to proving the second part of the proposition.

Observe that whenever \cref{alg:local_search_pass} changes its solution while processing element $u_i$, the value of this solution changes by
\begin{align*}
	f(u_i \mid S_{i - 1}) - f(u'_i \mid S_{i - 1} + u_i - u'_i)
	\geq{} &
	c \cdot f(u'_i : S_{i - 1}) - f(u'_i \mid S_{i - 1} + u_i - u'_i)\\
	\geq{} &
	(c - 1) \cdot f(u'_i \mid S_{i - 1} + u_i - u'_i)
	\geq
	0
	\enspace,
\end{align*}
where the second inequality holds by the submodularity of $f$ and the last inequality follows from the monotonicity of $f$. This implies that $f(S_i)$ is a non-decreasing function of $i$, and therefore,
\begin{align*}
	f(S_n)&{} - f(S_0)
	\geq
	\sum_{u_i \in B \cap (A \setminus S_0)} \mspace{-18mu} [f(u_i \mid S_{i - 1}) - f(u'_i \mid S_{i - 1} + u_i - u'_i)]\\
	\geq{} &
	\sum_{u_i \in B \cap (A \setminus S_0)} \mspace{-18mu} [f(u_i \mid S_{i - 1}) - f(u'_i \mid S_{i - 1} - u'_i)]
	\geq
	\sum_{u_i \in B \cap (A \setminus S_0)} \mspace{-18mu} [f(u_i \mid S_{i - 1}) - f(u'_i : S_{i - 1})]\\
	\geq{} &
	\sum_{u_i \in B \setminus S_0} f(u_i \mid S_{i - 1}) - \sum_{u_i \in B \cap (A \setminus S_0)} \mspace{-18mu} f(u'_i : S_{i - 1}) - c \cdot \sum_{u_i \in B \setminus A} f(u'_i : S_{i - 1})\\
	\geq{} &
	\sum_{u_i \in B \setminus S_0} f(u_i \mid S_{i - 1}) - c^2 \cdot \sum_{u_i \in B \cap (A \setminus S_0)} \mspace{-18mu} f(u'_i : S_{i - 1}) - c \cdot \sum_{u_i \in B \setminus A} f(u'_i : S_{i - 1})
	\enspace,
\end{align*}
where the second and third inequalities hold by the submodularity of $f$, the penultimate inequality holds since the elements of $B \setminus A$ where not added by \cref{alg:local_search_pass} to its solution, and the last inequality holds by the monotonicity of $f$. 

Using \cref{lem:adding_bound} and \cref{cor:removing_bound}, the previous inequality implies
\begin{align*}
	f(S_n) - f(&S_0)\\
	\geq{} &
	\{f(S_0 \cup B) + \tfrac{1}{c-1} \cdot f(S_0) - \tfrac{c}{c-1} \cdot f(S_n)\} - c \cdot \{f(S_n \mid \varnothing) - f(S_0 \mid S_0 \setminus B)\} \\
	={} &
	f(S_0 \cup B) + c \cdot f(\varnothing) + \tfrac{1}{c - 1} \cdot f(S_0) - \tfrac{c^2}{c - 1} \cdot f(S_n) + c \cdot f(S_0 \mid S_0 \setminus B)\\
	\geq{} &
	f(S_0 \cup B) + c \cdot f(\varnothing) + \tfrac{1}{c - 1} \cdot f(S_0) - \tfrac{c^2}{c - 1} \cdot f(S_n) + f(S_0 \mid S_0 \setminus B)
	\enspace,
\end{align*}
where the second inequality follows from monotonicity of $f$ and the fact that $c > 1$.
The first inequality of the proposition now follows by rearranging the last inequality (this can be verified by checking term by term).

To see that the second inequality of the proposition holds as well, we note that, by the submodulary of $f$,
\[
    f(B \mid S_0 \setminus B)
    =
     f(B \mid S_0) + f(B \cap S_0 \mid S_0 \setminus B)
    \geq
     f(B \mid S_0) + \sum_{u \in B \cap S_0} f(u \mid S_0 - u)
     \enspace.
     \qedhere
\]
\end{proof}

\subsection{Obtaining an Approximate Local Maximum} \label{ssc:local_search_multi}

As promised above, in this section we show that \cref{alg:local_search_pass} can be used to get an $\eps$-approximate local maximum. The algorithm we use to do that is given as \cref{alg:local_search_multi}, and it gets $\eps \in (0, 1)$ as a parameter.

\begin{algorithm}
\caption{\textsc{Multiple Local Search Passes} $(\eps)$} \label{alg:local_search_multi}
\begin{algorithmic}[1]
\State Find a base $T_0$ of $M$ using a single pass (by simply initializing $T_0$ to be the empty set, and then adding to it any elements that arrives and can be added to $T_0$ without violating independence in $M$).
\State Let $T_1$ be the output of \cref{alg:local_search_pass} when given $S_0 = T_0$ and $c = 2$. \label{line:OPT_estimation}
\For{$i = 2$ \textbf{to} $2 + \lceil 40\eps^{-2} \rceil$} \label{line:passes_loop}
	\State Let $T_i$ be the output of \cref{alg:local_search_pass} when given $S_0 = T_{i - 1}$ and $c = 1 + \eps / 2$.
	\If{$f(T_i) - f(T_{i - 1}) \leq \eps^2/10 \cdot f(T_1 \mid \varnothing)$}
		\State \Return{$T_{i - 1}$}. 
	\EndIf
\EndFor
\State Indicate failure if the execution of the algorithm has arrived to this point.
\end{algorithmic}
\end{algorithm}

Intuitively, \cref{alg:local_search_multi} works by employing the fact that every execution of \cref{alg:local_search_pass} increases the value of its input base $T_{i - 1}$ significantly, unless this input base is close to being a local maximum, and therefore, if the execution produces a base $T_i$ which is not much better than $T_{i - 1}$, then we know that $T_{i - 1}$ is an $\eps$-approximate local maximum. The following lemma proves this formally.
\begin{lemma}
If \cref{alg:local_search_multi} does not indicate a failure, then its output set $T$ obeys $f(B \mid T) + \sum_{u \in B \cap T} f(u \mid T - u) - f(T \mid \varnothing) < \eps \cdot f(\OPT \mid \varnothing)$. Note that since $B$ is an arbitrary base of $M$, the last inequality implies that $T$ is an $\eps$-approximate local maximum with result to $f$.
\end{lemma}
\begin{proof}
Since $T_1$ is a base of $M$, $f(T_1 \mid \varnothing) = f(T_1) - f(\varnothing) \leq f(\OPT) - f(\varnothing) = f(\OPT \mid \varnothing)$. This implies that when \cref{alg:local_search_multi} returns a set $T_{i - 1}$, then
\[
	f(T_{i}) - f(T_{i - 1}) \leq (\eps^2/10) \cdot f(\OPT \mid \varnothing)
	\enspace.
\]
Plugging this inequality and the fact that $f(T_{i} \mid \varnothing) \leq f(\OPT \mid \varnothing)$ (because $T_{i}$ is a base of $M$) into the guarantee of \cref{prop:single_pass} for the execution of \cref{alg:local_search_pass} that has created $T_{i}$ yields
\begin{align*}
	\eps \cdot f(\OPT \mid \varnothing)
	\geq{} &
	(\eps / 2) \cdot f(\OPT \mid \varnothing) + \eps(3\eps / 2 + 1)/5 \cdot f(\OPT \mid \varnothing)\\
	\geq{} &
	(\eps / 2) \cdot f(T_{i} \mid \varnothing) + \tfrac{3\eps/2 + 1}{\eps/2} \cdot [f(T_{i}) - f(T_{i - 1})]\\
	\geq{} &
	f(B \mid T_{i - 1}) + \sum_{u \in B \cap T_{i - 1}} \mspace{-18mu} f(u \mid T_{i - 1} - u) - f(T_{i - 1} \mid \varnothing)
	\enspace.
	\qedhere
\end{align*}
\end{proof}

One could image that it is possible for the value of the solution maintained by \cref{alg:local_search_multi} to increase significantly following every iteration of the loop starting on Line~\ref{line:passes_loop}, which will result in the algorithm indicating failure rather than ever returning a solution. However, it turns out that this cannot happen because the value of the solution of \cref{alg:local_search_multi} cannot exceed $f(\OPT)$, which implies a bound on the number of times this value can be increased significantly. This idea is formalized by the next two claims.
\begin{observation} \label{obs:opt_estimation}
$f(T_1 \mid \varnothing) \geq \tfrac{1}{5}f(\OPT \mid \varnothing)$.
\end{observation}
\begin{proof}
If we set $B = \OPT$, then by applying \cref{prop:single_pass} to the execution of \cref{alg:local_search_pass} on Line~\ref{line:OPT_estimation} of \cref{alg:local_search_multi}, we get
\begin{align*}
	f(T_1 \mid \varnothing) + 4[f(T_1) - f(T_0)]
	\geq{} &
	f(\OPT \mid T_0) + \sum_{u \in \OPT \cap T_0} \mspace{-9mu} f(u \mid T_0 - u) - f(T_0 \mid \varnothing)\\
	\geq{} &
	f(\OPT \mid T_0) - f(T_0 \mid \varnothing) \enspace,
\end{align*}
where the second inequality follows from the monotonicity of $f$. Since the leftmost side the last inequality is equal to $5f(T_1 \mid \varnothing) - 4f(T_0 \mid \varnothing)$, this inequality implies
\begin{align*}
	5f(T_1 \mid \varnothing)
	\geq{} &
	f(\OPT \mid T_0) + 3f(T_0 \mid \varnothing)
	=
	f(\OPT \cup T_0) + 2f(T_0) - 3f(\varnothing)\\
	\geq{} &
	f(\OPT) - f(\varnothing)
	=
	f(\OPT \mid \varnothing)
	\enspace,
\end{align*}
where the second inequality follows again from the monotonicity of $f$. The observation now follows by dividing the last inequality by $5$.
\end{proof}

\begin{corollary}
\cref{alg:local_search_multi} never indicates failure.
\end{corollary}
\begin{proof}
If $f(\OPT \mid \varnothing) = 0$, then the value of every base of $M$ according to $f$ is $f(\varnothing)$, which guarantees that \cref{alg:local_search_multi} returns $T_1$ during the first iteration of the loop starting on its \cref{line:passes_loop}. Therefore, we assume below that $f(\OPT \mid \varnothing) > 0$. Furthermore, assume towards a contradiction that \cref{alg:local_search_multi} indicates failure. By \cref{obs:opt_estimation}, this assumption implies that the value of the solution maintained by \cref{alg:local_search_multi} increases by at least $(\eps^2 / 10) \cdot f(T_1 \mid \varnothing) \geq \tfrac{\eps^2}{50} f(\OPT \mid \varnothing)$ after every iteration of the loop starting on Line~\ref{line:passes_loop}. Therefore, after all the $1 + \lceil 40\eps^{-2} \rceil$ iterations of this loop, the value of the solution of \cref{alg:local_search_multi} is at least
\[
	f(T_1) + (1 + \lceil 40\eps^{-2} \rceil) \cdot \tfrac{\eps^2}{50} f(\OPT \mid \varnothing)
	>
	f(\varnothing) + \tfrac{1}{5}f(\OPT \mid \varnothing) + \tfrac{4}{5} f(\OPT \mid \varnothing)
	=
	f(\OPT)
	\enspace,
\]
which is a contradiction since the solution of \cref{alg:local_search_multi} is always kept as a base of $M$.
\end{proof}

We now observe that \cref{alg:local_search_multi} has all the properties guaranteed by \cref{prop:local_max_adversarial}, which we repeat here for convenience. In particular, we note that \cref{alg:local_search_multi} can be implemented as a semi-streaming algorithm storing $O(k)$ elements because it needs to store at most two solutions at any given time.
\propLocalMaxAdversarial*

\subsection{Discussion of a Lower Bound by McGregor and Vu\texorpdfstring{~\cite{mcgregor2019better}}{}} \label{ssc:lower_bound_multipass}

McGregor and Vu~\cite{mcgregor2019better} showed that any data stream algorithm for the Maximum $k$-Coverage Problem (which is a special case of {\MSMM} in which $f$ is a coverage function and $M$ is a uniform matroid of rank $k$) that makes a constant number of passes must use $\Omega(m/k^2)$ memory to achieve $(1+\eps) \cdot (1-(1-1/k)^k)$-approximation with probability at least $0.99$, where $m$ is the number of sets in the input, and it is assumed that these sets are defined over a ground set of size $n = \Omega(\eps^{-2} k \log m)$. Understanding the implications of this lower bound for {\MSMM} requires us to handle two questions.
\begin{itemize}
	\item The first question is how the lower bound changes as a function of the number of passes. It turns out that when the number of passes is not dropped from the asymptotic expressions because it is considered to be a constant, the lower bound of McGregor and Vu~\cite{mcgregor2019better} on the space complexity becomes $\Omega(m/(pk^2))$, where $p$ is the number of passes done by the algorithm.
	\item The second question is about the modifications that have to be done to the lower bound when it is transferred from the Maximum $k$-Coverage Problem to {\MSMM}. Such modifications might be necessary because of input representation issues. However, as it turns out, the proof of the lower bound given by~\cite{mcgregor2019better} can be applied to {\MSMM} directly, yielding the same lower bound (except for the need to replace $m$ with the corresponding value in {\MSMM}, namely, $|\cN|$). Furthermore, McGregor and Vu~\cite{mcgregor2019better} had to use a very large ground set so that random sets will behave as one expects with high probability. When the objective function is a general submodular function, rather than a coverage function, it can be chosen to display the above-mentioned behavior of random sets, and therefore, $\eps$ can be set to $0$.
\end{itemize}

We summarize the above discussion in the following corollary.
\begin{corollary}[Corollary of McGregor and Vu~\cite{mcgregor2019better}]
For any $k \geq 1$, any $p$-pass data stream algorithm for {\MSMM} that achieves an approximation guarantee of $1-(1-1/k)^k \leq 1 - 1/e + 1/k$ with probability at least $0.99$ must use $\Omega(|\cN|/(pk^2))$ memory, and this is the case even when the matroid $M$ is restricted to be a uniform matroid of rank $k$.
\end{corollary}

\section{Approximate Local Maximum for Random Streams} \label{sec:random_local_search}

In this section we study {\MSMM} in random order streams by building on ideas from the analysis of Liu et al.~\cite{LRVZ20} for optimizing $f$ under a cardinality constraint. We begin with \cref{ssc:local_search_pass_randomized}, which analyzes and simplifies the single-pass local search algorithm of Shadravan~\cite{S20}. By applying this algorithm multiple times (in multiple passes), we are able to prove in \cref{ssc:local_search_multi_randomized} the following proposition (we note that \cref{ssc:local_search_multi_randomized} highly depends on \cref{sec:framework}, and should not be read before that section). \cref{prop:local_max_randomized} implies \cref{thm:multipass_random} by \cref{lem:local_maximum_good} and \cref{prop:equation_enough}.
\begin{restatable}{proposition}{propLocalMaxRandomized} \label{prop:local_max_randomized}
For every constant $\eps > 0$, there is a multi-pass semi-streaming algorithm that given an instance of {\MSMM}  with a matroid of rank $k$ stores $O(k / \eps)$ elements and makes $O(\eps^{-1} \log \eps^{-1})$ many passes. Assuming the order of the elements in the input stream is chosen uniformly at random in each pass, this algorithm outputs a solution $D$ such that $\bE[f(D \mid \varnothing)] \geq (1 - \eps) \cdot f(D' \mid \varnothing)$, where $D'$ is the $\eps$-approximate local maximum whose value with respect to $f$ is the smallest.
\end{restatable}

In \cref{ssc:matchoid} we observe that our single-pass algorithm from \cref{ssc:local_search_pass_randomized} can naturally be extented to $p$-matchoids. Then, we create a multi-pass algorithm based on this extended single-pass algorithm, which proves \cref{thm:local_optimum}. 

\subsection{Local Search for Matroid Constraints} \label{ssc:local_search_pass_randomized}

Intuitively, a local search algorithm should make a swap in its solution whenever this is beneficial. In the adversarial setting, one has to make a swap only when it is beneficial enough to avoid making too many negligible swaps (this is why $c > 1$ in \cref{alg:local_search_pass}). However, in the random order setting there is a better solution for this problem. Specifically, we (randomly) partition the input stream into \emph{windows} ($\alpha k$ contiguous chunks of the stream with expected size $n/(\alpha k)$ each for some parameter $\alpha > 1$), and then make the best swap within each window. Formally, our random partition is generated according to \cref{alg:partition}.

\begin{algorithm}[ht]
\caption{Partitioning of the input stream ($\alpha$) \label{alg:partition}}
\begin{algorithmic}[1]
\State Draw $|\cN|$ integers uniformly and independently from $1,2,\ldots, \alpha k$.
\For{$i = 1$ to $\alpha k$}
    \State Let $n_i \gets$ \# of integers equal to $i$.
    \State Let $t_i \gets \sum_{j = 1}^{i - 1} n_i$.
    \State Let $w_i \gets$  elements $t_i + 1$ to $t_i + n_i$ in $\cN$.
\EndFor
\State \Return $\{w_1, w_2, \ldots, w_{\alpha k}\}$.
\end{algorithmic}
\end{algorithm}

Our full single pass algorithm, which uses the partition defined by \cref{alg:partition}, is given as \cref{alg:matroid-basic}. The input for the algorithm includes the parameter $\alpha$ and some base $L_0$ of the matroid $\cM$. Additionally, during the execution of the algorithm, the set $L_i$ represents the current solution, and $H$ is the set of all elements that were added to this solution at some point. When processing window $w_i$, \cref{alg:matroid-basic} constructs a set $C_i$ of elements that can potentially be swapped into the solution. This set contains all the elements of the window plus some historical elements (the set $R_i$). The idea of using a set $H$ to store previously valuable elements is inspired from \cite{SMC19, LRVZ20}. Reintroducing previously seen elements allows us to give any element not in the solution a chance of being introduced into the solution in the future, which helps us avoid issues that result from the dependence that exists between the current solution and the set of elements in the current window.

\begin{algorithm}[ht]
\caption{{\sc MatroidStream}$(\alpha, L_0)$ \label{alg:matroid-basic}}
\begin{algorithmic}[1]
\State{Partition $\cN$ into windows $w_1, w_2, \dotsc, w_{\alpha k}$.}
\State{Let $H \gets \varnothing$.}
\For{$i=1$ to $\alpha k$}
    \State{Let $R_i$ be a random subset of $H$ including every $u \in H$ with probability $\frac{1}{\alpha k}$, inde-\hspace*{1cm}pendently.\label{line:R_selection}}
    \State{Let $C_i \gets w_{i} \cup R_i$}
    
    \State{Let $u^{\star}$ and $u_r^{\star}$ be elements maximizing $f(L_i - u_r^\star + u^\star)$ subject to the constraints: \hspace*{1cm}$u^{\star} \in C_i$, $u_r^{\star} \in L_i$ and $L_i - u_r^\star + u^\star \in \cI$ \label{algline:argmax}}.
    
    \If{$f(L_i) < f(L_i - u_r^{\star} + u^{\star})$}
        \State{Update $H \gets H + u^{\star}$.}
        \State{Let $L_{i+1} \gets L_i - u_r^{\star} + u^{\star}$.}
    \EndIf
\EndFor
\State \Return $L_{\alpha k}$.
\end{algorithmic}
\end{algorithm}

 Note that the number of elements stored by \cref{alg:matroid-basic} is $O(\alpha k)$, as this number is dominated by the size of the set $H$. For the same reason \cref{alg:matroid-basic} is a semi-streaming algorithm whenever $\alpha$ is constant.

\begin{definition}
\label{def:history}
Let $H_i$ denote the state of the set $H$ maintained by \cref{alg:matroid-basic} immediately after processing window $i$.
We define $\cH_i$ to be the set of all pairs $(u,j)$ such that element $u \in H_i$ was added to the solution while window $j$ was processed (i.e., $u \in H_i \cap w_j$). For convenience, sometimes we treat $\cH_i$ as a set of elements, and say that $u \in \cH_i$ if $u \in H_i$. 
\end{definition}
One can observe that $\cH_i$ encodes all the changes that the algorithm made to its state while processing the first $i$ windows because the element removed from the solution when $u$ is added is deterministic. Additionally, we note that different random permutations of the input and random coins in Line~\ref{line:R_selection} of \cref{alg:matroid-basic} may produce the same history, and we average over all of them in the analysis. 

The next lemma is from \cite{LRVZ20}. It captures the intuition that any element not selected by the algorithm still appears uniformly distributed in future windows, and bounds the probability with which this happens. The proof of this lemma can be found in \cref{sec:random_appendix}.
\begin{restatable}{lemma}{lemAtLeastProb}
\label{lem:at-least-prob}
Fix a history $\cH_{i-1}$ for some $ i\in [\alpha k]$.
For any element $u \in \cN \setminus \cH_{i-1}$, and any $i \leq j \leq \alpha k$, we have 
$Pr[u\in w_j \mid \cH_{i-1}] \geq 1/(\alpha k).$
\end{restatable}

Let $B$ an arbitrary base of $\cM$ (one can think of $B$ as an optimal solution because the monotonicity of $f$ guarantees that there is an optimal solution that is a base, but we sometimes need to consider other bases as $B$). We now define ``active'' windows, which are windows for which we can show a definite gain in our solution. Specifically, we show below that in any active window the value of the current solution $L$ increases roughly by $\frac{1}{k}(f(B) - 2f(L))$ in expectation, which yields an approximation ratio of $\frac{1}{2}(1-1/e^2)$ after $\alpha k$ windows have been processed in one pass because we expect roughly one in every $\alpha$ windows to be active. 
\begin{definition}
\label{def:active-window}
For window $w_i$, let $p_u^i$ be the probability that $u \in w_i$ conditioned $\cH_{i-1}$. Define the active set $A_i$ of $w_i$ to be the union of $R_i$ and a set obtained by sampling each element $u \in w_i$ with probability $1/(\alpha k p_u^i)$. We call $w_i$ an \textbf{active window} if $|B \cap A_i| \geq 1$.
\end{definition}
Note that the construction of active sets in \Cref{def:active-window} is valid as \Cref{lem:at-least-prob} guarantees that $1/(\alpha k p_e^i)$ is a valid probability (i.e., it is not more than $1$). More importantly, the active set $A_i$ includes every element of $\cN$ with probability exactly $1/(\alpha k)$, even conditioned on the history $\cH_{i-1}$; which implies that, since each element appears in $A_i$ independently, a window is active with probability $(1-1/(\alpha k))^k \geq 1 - e^{-1/\alpha} \approx 1/\alpha$ conditioned on any such history.

Before proving the guarantees of \Cref{alg:matroid-basic}, we also need to introduce the following known lemmata.
\begin{lemma}[Lemma~2.2 of~\cite{feige2011maximizing}]
\label{lem:subsample-matroid}
Let $g\colon 2^\cN \rightarrow \bR$ be a submodular function. Further, let $R$ be a random subset of $T\subseteq \cN$ in which every element occurs with probability $p$ (not necessarily independently). Then, $\bE[g(R)] \geq p \cdot g(T) + (1 - p) \cdot g(\varnothing)$.
\end{lemma}

\begin{lemma}[Follows, for example, from Corollary~39.12a of~\cite{schrijver2003combinatorial}]
\label{lem:matroid-bijection}
If $S$ and $T$ are two bases of a matroid $\cM = (\cN,\cI)$, then there exists a bijection $h\colon T \rightarrow S$ such that for every $u \in T$,  $T - h(u) + u \in \cI$. Furthermore, for every element $u \in S \cap T$, $h(u) = u$.
\end{lemma}

Let $\cA_i$ denote the event that window $i$ is active. The following lemma lower bounds the increase in the value of the solution of \cref{alg:matroid-basic} in an active window.
\begin{lemma}
\label{lem:matroid-gain}
For every integer $0 \leq i < \alpha k$,
\begin{align*}
    \bE[f(L_{i+1}) - f(L_i) \mid \cH_{i}, \cA_{i+1}]
    \geq{} &
    \tfrac{1}{k} \bE\left[f(B \mid L_i) + \sum_{u \in B \cap L_i} \mspace{-9mu} f(u \mid L_i - u) - f(L_i \mid \varnothing) ~\middle|~ \cH_i \right]\\
    \geq{} &
    \tfrac{1}{k} \bE[f(B) - 2f(L_i) \mid \cH_i]
    \enspace.
\end{align*}
Moreover, the above inequality holds even when $B$ is a random base as long as it is deterministic when conditioned on any given $\cH_i$.
\end{lemma}

\begin{proof}
Let us apply \cref{lem:matroid-bijection} with $S = L_i$ and $T = B$ to get a bijection $h\colon B \to L_i$ with the properties specified in the lemma. Let $b$ denote a uniformly random element of $A_{i + 1} \cap B$. Since every element of $B$ belongs to $A_{i + 1}$ with probability $1 / (\alpha k)$, independently, even conditioned on $\cH_i$, and the event $\cA_{i + 1}$ simply excludes the possibility that $A_{i + 1} \cap B$ is empty, we get that $b$ is a uniformly random element of $B$ when conditioned on $\cH_{i + 1}$ and $\cA_{i + 1}$. Furthermore, since $h$ is a bijection, $h(b)$ is a uniformly random element of $L_i$ under the same conditioning, which implies that every element of $L_i$ appears in $L_i - h(b)$ with probability $1 - 1/k$.

Given the above observations, we get
\begin{align*}
\bE[ f(L_{i+1}\mspace{-2mu}&\mspace{2mu}) \mid \cH_i, \cA_{i+1}] \geq \bE [f(L_i - h(b) + b) \mid \cH_i, \cA_{i+1}] \\
     ={} & \bE [f(L_i - h(b)) \mid \cH_i, \cA_{i+1}] + \bE [f(b \mid L_i - h(b)) \mid \cH_i, \cA_{i+1}] \\
     ={} &
     \bE [f(L_i - h(b)) \mid \cH_i, \cA_{i+1}] + \bE [\tfrac{1}{k}{\textstyle \sum_{u \in B}} f(u \mid L_i - h(u)) \mid \cH_i, \cA_{i+1}]\\
     \geq{} & (1 - 1/k) \cdot \bE[f(L_i) \mid \cH_i, \cA_{i+1}] + \tfrac{1}{k} f(\varnothing) + \bE [\tfrac{1}{k}{\textstyle\sum_{u \in B}} f(u \mid L_i - h(u)) \mid \cH_i, \cA_{i+1}] \\
     ={} & (1 - 1/k) \cdot \bE[f(L_i) \mid \cH_i]+ \tfrac{1}{k} f(\varnothing) + \tfrac{1}{k} \cdot \bE [{\textstyle\sum_{u \in B}} f(u \mid L_i - h(u)) \mid \cH_i]
\end{align*}
where the second inequality follows from \cref{lem:subsample-matroid}, and the last equality holds since $L_i$ and $B$ are deterministic given $\cH_i$. Using the submodularity of $f$, and recalling that $h(u) = u$ for $u \in L_i \cap B$, we can now lower bound the argument of the second expectation on the rightmost side of the last inequality as follows.
\begin{align*}
    \sum_{u \in B} f(u \mid L_i - h(u))
    \geq &
    \sum_{u \in B \setminus L_i} \mspace{-9mu} f(u \mid L_i) + \sum_{u \in B \cap L_i} \mspace{-9mu} f(u \mid L_i - u)\\
    \geq{} &
    f(B \mid L_i) + \sum_{u \in B \cap L_i} \mspace{-9mu} f(u \mid L_i - u)
    \enspace.
\end{align*}
The first inequality of the lemma follows by plugging the last inequality into the previous one, and rearranging. Furthermore, the second inequality of the lemma holds since $f(L_i \mid \varnothing) \leq f(L_i)$ and the monotonicity of $f$ implies
\[
    f(B \mid L_i) + \sum_{u \in B \cap L_i} \mspace{-9mu} f(u \mid L_i - u)
    \geq
    f(B \mid L_i)
    \geq
    f(B) - f(L_i)
    \enspace.
    \qedhere
\]
\end{proof}

Technically, Lemma~\ref{lem:matroid-gain} suffices to prove our results. However, it is useful to also prove the following theorem, which reproves a result due to~\cite{S20b}. To understand this theorem we need to make two observations.
\begin{itemize}
    \item We would like to chose $\alpha$ on the order of $1/\eps$ to guarantee that the error term diminishes with $\eps$. However, we also need to guarantee that $\alpha k$ is integral (which is necessary for \cref{alg:matroid-basic}). The value we choose for $\alpha$ in \cref{thm:matroid-approx} is designed to satisfy these two requirements.
    \item As given, \cref{alg:matroid-basic} requires a base $L_0$ of $\cM$ as input. Since we do not care about the value of this base in \cref{thm:matroid-approx} (we care about it in the next section), we can mimic having such a base using the following idea. First, we pretend to add $k$ dummy elements to the ground set such that (i) the dummy elements do not affect the value of any set according to $f$; and (ii) a set that includes dummy elements is independent in the matroid constraint if it is independent when the dummy elements are removed, and its original size before the removal is at most $k$ (see~\cite{buchbinder2014submodular} for a proof that adding such dummy elements does not affect the properties we assume for the objective function and constraint). Once the dummy element are added, we can choose $L_0$ to simply be the base consisting of the $k$ dummy elements.
\end{itemize}

\begin{theorem}
\label{thm:matroid-approx}
For every $\eps \in [0, 1]$, setting $\alpha = \lceil k / \eps \rceil / k$ and initializing $L_0$ to be the set of dummy elements as described above makes \cref{alg:matroid-basic} a semi-streaming algorithm guaranteeing $\frac{1}{2}(1 - 1/e^2) - O(\eps)$ approximation and storing $O(k/\eps)$ elements.
\end{theorem}

\begin{proof}
Recall that $1 - e^{-1/\alpha} \leq \prob{\cA_{i+1}} = 1-(1-1/(\alpha k))^k \leq 1/\alpha$. Thus, for every integer $0 \leq i < \alpha k$,
\begin{align*}
\bE[ f(L_{i+1}) &{} \mid \cH_i] = \Pr[\cA_{i+1}] \cdot \bE[ f(L_{i+1}) \mid \cH_i, \cA_{i+1}] + \Pr[\neg{\cA_{i+1}}] \cdot \bE[ f(L_{i+1}) \mid \cH_i, \neg{\cA_{i+1}}] \\
&\geq \Pr[\cA_{i+1}] \cdot \bE[ f(L_{i+1}) \mid \cH_i, \cA_{i+1}] + \Pr[\neg{\cA_{i+1}}] \cdot \bE[f(L_i) \mid \cH_i] \\
&\geq \Pr[\cA_{i+1}] \cdot \left((1-\tfrac{2}{k}) \cdot \bE[f(L_i) \mid \cH_i] + \tfrac{1}{k} f(\OPT)\right) + \Pr[\neg{\cA_{i+1}}] \cdot \bE[ f(L_i) \mid \cH_i] \\
&\geq \left(1-\frac{2}{\alpha k}\right) \cdot \bE[f(L_i) \mid \cH_i] + \tfrac{1}{k}(1-e^{-1/\alpha}) \cdot f(\OPT)
\end{align*}
where the first inequality follows from the facts that the algorithm only increases the value of its solution and $\cH_i$ completely determines $L_i$, and the second inequality follows from \Cref{lem:matroid-gain} by choosing $B = \OPT$. The law of total expectation allows us to remove the conditioning on $\cH_i$ from both sides of the last inequality, which yields (by repeated applications of the last inequality and using the fact that $f(L_0) \geq 0$) the inequality
\begin{align*}
    \bE[f(L_i)]
    \geq{} &
    \sum_{j = 1}^i \left(1 - \frac{2}{\alpha k}\right)^{i - j} \cdot \tfrac{1}{k}(1-e^{-1/\alpha}) \cdot f(\OPT)\\
    \geq{} &
    \frac{1 - (1 - 2/(\alpha k))^i}{1 - (1 - 2/(\alpha k))} \cdot \tfrac{1}{k}(1-e^{-1/\alpha}) \cdot f(\OPT)\\
    ={} &
    \tfrac{\alpha}{2}(1-e^{-1/\alpha}) \cdot \left(1 - (1 - 2/(\alpha k))^i\right) \cdot f(\OPT)
    \enspace.
\end{align*}
To simplify this inequality, we observe that $1 - 2 / (\alpha k) \leq e^{-2 / (\alpha k)}$ and $\alpha(1 - e^{-1/\alpha}) \geq \alpha(1 / \alpha - 1 / \alpha^2) = 1 - 1 / \alpha$, which yields
\[
    \bE[f(L_i)]
    \geq
    \tfrac{1}{2}(1 - e^{-2i / (\alpha k )} - 1/ \alpha) \cdot f(\OPT)
    =
    \tfrac{1}{2}(1 - e^{-2i / (\alpha k )} - O(\eps)) \cdot f(\OPT)
    \enspace.
\]
The theorem now follows by plugging $i = \alpha k$ into this inequality.
\end{proof}

\subsection{Getting Almost as Good Solution as an Approximate Local Maximum \label{ssc:local_search_multi_randomized}}

We consider in this section a multi-pass algorithm (given as \cref{alg:local_search_multi_random}) obtained by running \cref{alg:matroid-basic} $\Theta(\eps^{-1} \log \eps^{-1})$ times, feeding the output of each execution as the input for the next execution. As promised above, we show that the algorithm obtained in this way outputs a solution whose expected value is almost as good as some $\eps$-approximate local maximum. \Cref{alg:local_search_multi_random} gets $\eps \in (0, 1/2]$ as a parameter.

\begin{algorithm}
\caption{\textsc{Multiple Local Search Passes for Random Streams} $(\eps)$} \label{alg:local_search_multi_random}
\begin{algorithmic}[1]
\State Let $r = \lceil 2 \eps^{-1} \ln \eps^{-1} \rceil$.
\State Find a base $T_0$ of $M$ using a single pass.
\For{$j = 1$ \textbf{to} $r$}
	\State Let $T_j$ be the output of \cref{alg:local_search_pass} when given $L_0 = T_{j - 1}$ and $\alpha = \lceil k / \eps \rceil / k$.
\EndFor
\State \Return $T_{r}$.
\end{algorithmic}
\end{algorithm}

We begin the analysis of \cref{alg:local_search_multi_random} by showing that the expected value of the solution of \cref{alg:matroid-basic} increases significantly in every window as long as this solution is not an approximate local maximum.
\begin{observation} \label{obs:window_increase_non_local}
Consider some window $w_i$ in an execution of \cref{alg:matroid-basic} done within \cref{alg:local_search_multi_random}, and let $\cE$ be the event that the solution $L_i$ of the algorithm at the beginning this window is not an $\eps$-approximate local maximum, then
\[
    \bE[f(L_{i + 1}) - f(L_i) \mid \cE] \geq \frac{\eps}{2\alpha k} \cdot f(\OPT \mid \varnothing)
    \enspace.
\]
\end{observation}
\begin{proof}
Fix an history $\cH_i$ that implies $\cE$ (since $\cH_i$ completely determines $L_i$, it also determines $\cE$). By the law of total expectation, the lemma will follow if we can prove
\[
    \bE[f(L_{i + 1}) - f(L_i) \mid \cH_i] \geq \eps \cdot f(\OPT \mid \varnothing)
    \enspace.
\]
Therefore, in the rest of this proof we concentrate on proving this inequality.

Since $L_i$ is not an $\eps$-local maximum under $\cH_i$, we can choose $B$ to be a base such that
\[
    f(L_i \mid \varnothing) \leq f(B \mid L_i) + \sum_{B \cap L_i} f(u \mid L_i - u) - \eps \cdot f(OPT \mid \varnothing)
    \enspace.
\]
Then, \cref{lem:matroid-gain} implies
\begin{align*}
    \bE[f(L_{i+1}) - f(L_i) \mid \cH_{i}, \cA_{i+1}]
    \geq{} &
    \tfrac{1}{k} \bE\left[f(B \mid L_i) + \sum_{u \in B \cap L_i} \mspace{-9mu} f(u \mid L_i - u) - f(L_i \mid \varnothing) ~\middle|~ \cH_i \right]\\
    \geq{} &
    \tfrac{1}{k}\bE[\eps \cdot f(\OPT \mid \varnothing) \mid \cH_i]
    =
    \frac{\eps}{k} \cdot f(\OPT \mid \varnothing)
    \enspace.
\end{align*}

We are now ready to prove the observation. By the law of total expectation and the fact that the value of the solution of \cref{alg:matroid-basic} never decreases,
\begin{align*}
    \bE[f(L_{i+1}) - f(L_i) \mid \cH_{i}]
    \geq{} &
    \Pr[A_i \mid \cH_i] \cdot \bE[f(L_{i+1}) - f(L_i) \mid \cH_{i}, \cA_i]\\
    \geq{} &
    (1 - e^{-1/\alpha}) \cdot (\eps / k) \cdot f(\OPT \mid \varnothing)
    \geq
    \frac{\eps(1 - 1/\alpha)}{\alpha k} \cdot f(\OPT \mid \varnothing)\\
    \geq{} &
    \frac{\eps(1 - \eps)}{\alpha k} \cdot f(\OPT \mid \varnothing)
    \geq
    \frac{\eps}{2\alpha k} \cdot f(\OPT \mid \varnothing)
    \enspace,
\end{align*}
where the last inequality holds since $\eps \leq 1/2$.
\end{proof}

Let $D'$ be an $\eps$-approximate local maximum whose value according to $f$ is minimal among all $\eps$-approximation local maxima.
\begin{lemma}
The output $T_r$ of \cref{alg:local_search_multi_random} obeys $\bE[f(T_r \mid \varnothing)] \geq (1 - \eps) \cdot f(D' \mid \varnothing)$.
\end{lemma}
\begin{proof}
Consider the setting described in \cref{obs:window_increase_non_local}. By a Markov like argument, the probability of the event $\cE$ is at least $1 - \bE[f(L_i \mid \varnothing)]/f(D' \mid \varnothing)$ because the event $\cE$ happens whenever $f(L_i \mid \varnothing) < f(D' \mid \varnothing)$. Therefore,
\begin{align*}
    \bE[f(L_{i+1}) - f(L_i)]
    \geq{} &
    \Pr[\cE] \cdot \bE[f(L_{i+1}) - f(L_i) \mid \cE]\\
    \geq{} &
    \max\left\{0, \left(1 - \frac{\bE[f(L_i \mid \varnothing)]}{f(D' \mid \varnothing)}\right)\right\} \cdot \frac{\eps}{2\alpha k} \cdot f(OPT \mid \varnothing)\\
    \geq{} &
    \max\left\{0, \left(1 - \frac{\bE[f(L_i \mid \varnothing)]}{f(D' \mid \varnothing)}\right)\right\} \cdot \frac{\eps}{2\alpha k} \cdot f(D' \mid \varnothing)\\
    \geq{} &
    \frac{\eps}{2\alpha k} \cdot \{f(D' \mid \varnothing) - \bE[f(L_i \mid \varnothing)]\}
    \enspace,
\end{align*}
where the penultimate inequality holds since the inequality $f(\OPT) > f(D')$ follows from the fact that $\OPT$ is an optimal solution with respect to $f$. Rearranging this inequality yields,
\[
    f(D' \mid \varnothing) - \bE[f(L_{i + 1} \mid \varnothing)]
    \leq
    \left(1 - \frac{\eps}{2\alpha k}\right) \cdot \{f(D' \mid \varnothing) - \bE[f(L_{i} \mid \varnothing)]\}
    \enspace.
\]

The above inequality applies to every window in every one of the $r$ executions of \cref{alg:matroid-basic} that are used by \cref{alg:local_search_multi_random}. Since there are $r \alpha k$ such windows in all these executions of \cref{alg:matroid-basic}, combining the inequalities corresponding to all them yields
\begin{align*}
    f(D' \mid \varnothing) - \bE[f(T_{r} \mid \varnothing)]
    \leq{} &
    \left(1 - \frac{\eps}{2\alpha k}\right)^{r k \alpha} \cdot \{f(D' \mid \varnothing) - \bE[f(T_0) \mid \varnothing]\}\\
    \leq{} &
    e^{-\eps r / 2} \cdot \{f(D' \mid \varnothing) - \bE[f(T_0) \mid \varnothing]\}
    \leq
    e^{-\eps r / 2} \cdot f(D' \mid \varnothing)\\
    \leq{} &
    e^{-\ln \eps^{-1}} \cdot f(D' \mid \varnothing)
    =
    \eps \cdot f(D' \mid \varnothing)
    \enspace,
\end{align*}
where the penultimate inequality follows from the monotonicity of $f$. The lemma now follows by rearranging the last inequality.
\end{proof}

The last lemma completes the proof of \cref{prop:local_max_randomized} since \cref{alg:local_search_multi_random} uses $O(\eps^{-1} \log \eps^{-1})$ passes (one pass for each execution of \cref{alg:matroid-basic}) and stores only a single solution in addition to the $O(\alpha k) = O(k/\eps)$ elements stored by each execution of \cref{alg:matroid-basic}.\footnote{A technical issue is that we assume in the analysis of \cref{alg:local_search_multi_random} that $\eps \leq 1/2$. However, this assumption can be dropped by simply replacing $\eps$ with $1/2$ at the beginning of the algorithm if $\eps$ happens to be larger.}

\subsection{Extending to \texorpdfstring{$p$}{p}-Matchoid Constraints} \label{ssc:matchoid}

In this section we prove \cref{thm:local_optimum}, which we repeat here for convenience.
\thmLocalOptimum*

\Cref{alg:matchoid-basic} is a generalization of \cref{alg:matroid-basic} for a matchoid constraint. The one key difference between the algorithms is that adding an element $u$ to a solution $L_i$ may cause the removal of up to $p$ elements because $u$ can conflict with at most one element in each one of the $p$ matroids it is a member of. Additionally, while \cref{alg:matchoid-basic} requires the input set $L_0$ to be independent in the matchoid constraint, it does not require it to be a base (unlike \cref{alg:matroid-basic}, which does require that).

\begin{algorithm}[ht]
\caption{$p$-{\sc MatchoidStream}$(\alpha, L_0)$ \label{alg:matchoid-basic}}
\begin{algorithmic}[1]
\State{Partition $\cN$ into windows $w_1, w_2, \dotsc, w_{\alpha k}$.}
\State{Let $H \gets \varnothing$.}
\For{$i=1$ to $\alpha k$}
    \State{Let $R_i$ be a random subset of $H$ including every $u \in H$ with probability $\frac{1}{\alpha k}$, inde-\hspace*{1cm}pendently.}
    \State{Let $C_i \gets w_{i} \cup R_i$}
    
    \State{Let $u^{\star}$ and $S_r^{\star}$ be element and set, respectively, maximizing $f(L_i \setminus S^\star_r + u_r)$ subject \hspace*{1cm}to the constraints: $u^{\star} \in C_i$, $S_r^{\star} \subseteq L_i$ and $L_i \setminus S_r^\star + u^\star \in \cI$}.
    
    \If{$f(L_i) < f(L_i \setminus S_r^{\star} + u^{\star})$}
        \State{Update $H \gets H + u^{\star}$.}
        \State{Let $L_{i+1} \gets L_i \setminus S_r^{\star} + u^{\star}$.}
    \EndIf
\EndFor
\State \Return $L_{\alpha k}$.
\end{algorithmic}
\end{algorithm}

The analysis of \cref{alg:matchoid-basic} is identical to the analysis of \cref{alg:matroid-basic} up to (but excluding) \cref{lem:matroid-gain}. Therefore, we begin by proving the analog of the last lemma given below as \cref{lem:matchoid-gain}.

\begin{lemma}[Lemma~2.2 of~\cite{BFNS14}]
\label{lem:super-subsample-matroid}
Let $g\colon 2^\cN \rightarrow \nnR$ be a non-negative submodular function. Further, let $R$ be a random subset of $T\subseteq \cN$ in which every element occurs with probability at least $p$ (not necessarily independently). Then, $\bE[g(R)] \geq p \cdot g(T)$.
\end{lemma}

\begin{lemma}
\label{lem:matchoid-gain}
For every integer $0 \leq i < \alpha k$,
\[
    \bE[f(L_{i+1}) - f(L_i) \mid \cH_{i}, \cA_{i+1}]
    \geq
    \tfrac{1}{k} \bE[f(B) - (p + 1) \cdot f(L_i) \mid \cH_i]
    \enspace.
\]
\end{lemma}

\begin{proof}
For every integer $1 \leq \ell \leq q$, we would like to apply \cref{lem:matroid-bijection} with $S = B \cap \cN_\ell$ and $T = L_i \cap \cN_\ell$ to get a bijection $h_\ell \colon B \to L_i$ with the properties specified in the lemma with respect to the matroid $M_i$. This cannot be immediately done because $L_i \cap \cN_\ell$ and $B \cap \cN_\ell$. However, if we extend $L_i \cap \cN_\ell$ and $B \cap \cN_\ell$ to bases of $M_i$ in an arbitrary way, and then apply \cref{lem:matroid-bijection}, then we can get an injective function $h_\ell \colon (B \cap \cN_\ell) \to \cN_\ell$ such that $L_i \cap \cN_\ell - h(u) + u \in \cI$ for every element $u \in B \cap \cN_\ell$. 

Let $b$ denote now a uniformly random element of $A_{i + 1} \cap B$. Since every element of $B$ belongs to $A_{i + 1}$ with probability $1 / (\alpha k)$, independently, even conditioned on $\cH_i$, and the event $\cA_{i + 1}$ simply excludes the possibility that $A_{i + 1} \cap B$ is empty, we get that $b$ is a uniformly random element of $B$ when conditioned on $\cH_{i + 1}$ and $\cA_{i + 1}$. Furthermore, since $h_\ell$ is an injective function for every $\ell \in [q]$ and every element $u \in \cN$ belongs to $\cN_\ell$ for at most $p$ different values of $\ell$, the probability that $\bigcup_{\ell : b \in \cN_\ell} \{h(b)\}$ contains some element $u \in \cN$ is at most $p / |B|$. Therefore, if we denote $U(b) \coloneqq \bigcup_{\ell : b \in \cN_\ell}$, then every element of $L_i$ appears in $L_i \setminus U(b)$ with probability at least $1 - p/|B|$.

Since $L_i \setminus U(b) + b$ is independent in the matchoid constraint, the above observations imply
\begin{align*}
\bE[ f(L_{i+1}) \mid \cH_i&{}, \cA_{i+1}] \geq \bE [f(L_i \setminus U(b) + b) \mid \cH_i, \cA_{i+1}] \\
     ={} & \bE [f(L_i \setminus U(b)) \mid \cH_i, \cA_{i+1}] + \bE [f(b \mid L_i \setminus U(b)) \mid \cH_i, \cA_{i+1}] \\
     ={} &
     \bE [f(L_i \setminus U(b)) \mid \cH_i, \cA_{i+1}] + \bE [\tfrac{1}{|B|}{\textstyle \sum_{u \in B}} f(u \mid L_i \setminus U(u)) \mid \cH_i, \cA_{i+1}]\\
     \geq{} & (1 - p/|B|) \cdot \bE[f(L_i) \mid \cH_i, \cA_{i+1}] + \bE [\tfrac{1}{|B|}{\textstyle\sum_{u \in B}} f(u \mid L_i \setminus U(u)) \mid \cH_i, \cA_{i+1}] \\
     ={} & (1 - p/|B|) \cdot \bE[f(L_i) \mid \cH_i]+ \tfrac{1}{|B|} \cdot \bE [{\textstyle\sum_{u \in B}} f(u \mid L_i \setminus U(u)) \mid \cH_i]
\end{align*}
where the second inequality follows from \cref{lem:super-subsample-matroid}, and the last equality holds since $L_i$ and $B$ are deterministic given $\cH_i$. By the monotonicity and submodularity of $f$, we can lower bound the argument of the second expectation on the rightmost side of the last inequality as follows.
\[
    \sum_{u \in B} f(u \mid L_i \setminus U(u))
    \geq
    \sum_{u \in B \setminus L_i} \mspace{-9mu} f(u \mid L_i \setminus U(u))
    \geq
    \sum_{u \in B \setminus L_i} \mspace{-9mu} f(u \mid L_i)
    \geq
    f(B \mid L_i)
    \geq
    f(B) - f(L_i)
    \enspace.
\]

Plugging the last inequality into the previous one, we get
\[
    \bE[ f(L_{i+1}) - f(L_i) \mid \cH_i, \cA_{i+1}]
    \geq
    \tfrac{1}{|B|}\bE[f(B) - (p + 1)f(L_i) \mid \cH_i]
    \enspace.
\]
If $\bE[f(B) - (p + 1)f(L_i) \mid \cH_i] \geq 0$ then this inequality implies the lemma since the size of the independent set $B$ cannot exceed the rank $k$ of the matchoid. Otherwise, the lemma holds since \cref{alg:matchoid-basic} guarantees $f(L_{i + 1}) \geq f(L_i)$.
\end{proof}

The last lemma allows us to reprove also the following result due to~\cite{S20b}, which is a generalisation of \cref{thm:matroid-approx}.
\begin{theorem}
\label{lem:matchoid-approx}
For every $\eps \in [0, 1]$, setting $\alpha = \lceil k / \eps \rceil / k$ and initializing $L_0$ to be the empty set makes \cref{alg:matchoid-basic} a semi-streaming algorithm guaranteeing $\frac{1}{p + 1}(1 - 1/e^{p + 1}) - O(\eps)$ approximation and storing $O(k/\eps)$ elements.
\end{theorem}

\begin{proof}
Repeating the initial stages of the proof of \cref{thm:matroid-approx}, but using \cref{lem:matchoid-gain} instead of \cref{lem:matroid-gain}, we can get, for every integer $0 \leq i < \alpha k$,
\[
    \bE[ f(L_{i+1}) \mid \cH_i]
    \geq
    \left(1-\frac{p + 1}{\alpha k}\right) \cdot \bE[f(L_i) \mid \cH_i] + \tfrac{1}{k}(1-e^{-1/\alpha}) \cdot f(\OPT)
    \enspace.
\]
The law of total expectation allows us to remove the conditioning on $\cH_i$ from both sides of the last inequality, which yields (by repeated applications of the inequality and observing that $f(L_0) \geq 0$) the inequality
\begin{align*}
    \bE[f(L_i)]
    \geq{} &
    \sum_{j = 1}^i \left(1 - \frac{p + 1}{\alpha k}\right)^{i - j} \cdot \tfrac{1}{k}(1-e^{-1/\alpha}) \cdot f(\OPT)\\
    \geq{} &
    \frac{1 - (1 - (p + 1)/(\alpha k))^i}{1 - (1 - (p + 1)/(\alpha k))} \cdot \tfrac{1}{k}(1-e^{-1/\alpha}) \cdot f(\OPT)\\
    ={} &
    \tfrac{\alpha}{p + 1}(1-e^{-1/\alpha}) \cdot \left(1 - (1 - (p + 1)/(\alpha k))^i\right) \cdot f(\OPT)
    \enspace.
\end{align*}
To simplify this inequality, we observe that $1 - (p + 1) / (\alpha k) \leq e^{-(p + 1) / (\alpha k)}$ and $\alpha(1 - e^{-1/\alpha}) \geq \alpha(1 / \alpha - 1 / \alpha^2) = 1 - 1 / \alpha$, which yields
\begin{equation} \label{eq:matchoid_gain_multi_window}
    \bE[f(L_i)]
    \geq
    \tfrac{1}{p + 1}(1 - e^{-i(p + 1) / (\alpha k )} - 1/\alpha) \cdot f(\OPT)
    =
    \tfrac{1}{p + 1}(1 - e^{-i(p + 1) / (\alpha k )} - O(\eps)) \cdot f(\OPT)
    \enspace.
\end{equation}
The theorem now follows by plugging $i = \alpha k$ into this inequality.
\end{proof}

We can extend \cref{alg:matchoid-basic} into a multi-pass algorithm in the same way in which \cref{alg:matroid-basic} is extended into the multi-pass algorithm \cref{alg:local_search_multi_random} in \cref{ssc:local_search_multi_randomized}.\footnote{A technicality to consider is that, for a matchoid constraint, one cannot use the first pass to construct a base $T_0$. However, this is not an issue since \cref{alg:matchoid-basic} can get any independent set as $L_0$, which means that it is fine to simply set $T_0 \gets \varnothing$.} If this is done for $r$ passes, then the resulting algorithm has $r\alpha k$ windows instead of the $\alpha k$ windows of a single pass. Therefore, the expected value of the solution obtained at the end of $r$ passes is given by plugging $r \alpha k$ into Inequality~\eqref{eq:matchoid_gain_multi_window}, yielding a value of
\[
    \tfrac{1}{p + 1}(1 - e^{-r(p + 1)} - O(\eps)) \cdot f(\OPT)
    \enspace.
\]
Therefore, if we choose the number of passes $r$ to be $\lceil (p + 1)^{-1} \log \eps^{-1} \rceil$, the approximation ratio of the multi-pass algorithm becomes
\[
    \tfrac{1}{p + 1}(1 - e^{-\log \eps^{-1}} - O(\eps))
    =
    \frac{1 - \eps}{p + 1} - O(\eps)
    =
    \frac{1}{p + 1} - O(\eps)
    \enspace,
\]
which completes the proof of the second part of \cref{thm:local_optimum}. The first part of the theorem then follows because every matroid is a $1$-matchoid, and vice versa.

\bibliographystyle{plainurl}
\bibliography{combined}

\begin{thebibliography}{10}

\bibitem{SMC19}
Shipra Agrawal, Mohammad Shadravan, and Cliff Stein.
\newblock Submodular secretary problem with shortlists.
\newblock {\em CoRR}, abs/1809.05082:1:1--1:19, 2018.
\newblock URL: \url{http://arxiv.org/abs/1809.05082}, \href
  {http://arxiv.org/abs/1809.05082} {\path{arXiv:1809.05082}}, \href
  {https://doi.org/10.4230/LIPIcs.ITCS.2019.1}
  {\path{doi:10.4230/LIPIcs.ITCS.2019.1}}.

\bibitem{badanidiyuru2014streaming}
Ashwinkumar Badanidiyuru, Baharan Mirzasoleiman, Amin Karbasi, and Andreas
  Krause.
\newblock Streaming submodular maximization: Massive data summarization on the
  fly.
\newblock In {\em Proceedings of the 20th ACM Conference on Knowledge Discovery
  and Data Mining (KDD)}, pages 671--680, 2014.

\bibitem{badanidiyuru2014fast}
Ashwinkumar Badanidiyuru and Jan Vondr\'{a}k.
\newblock Fast algorithms for maximizing submodular functions.
\newblock In {\em Proceedings of the 25th Annual ACM-SIAM Symposium on Discrete
  Algorithms (SODA)}, SODA '14, pages 1497--1514, 2014.

\bibitem{buchbinder2019constrained}
Niv Buchbinder and Moran Feldman.
\newblock Constrained submodular maximization via a nonsymmetric technique.
\newblock {\em Mathematics of Operations Research}, 44(3):988--1005, 2019.
\newblock \href {https://doi.org/10.1287/moor.2018.0955}
  {\path{doi:10.1287/moor.2018.0955}}.

\bibitem{buchbinder2014submodular}
Niv Buchbinder, Moran Feldman, Joseph Naor, and Roy Schwartz.
\newblock Submodular maximization with cardinality constraints.
\newblock In Chandra Chekuri, editor, {\em Proceedings of the Twenty-Fifth
  Annual {ACM-SIAM} Symposium on Discrete Algorithms (SODA)}, pages 1433--1452.
  {SIAM}, 2014.
\newblock \href {https://doi.org/10.1137/1.9781611973402.106}
  {\path{doi:10.1137/1.9781611973402.106}}.

\bibitem{BFNS14}
Niv Buchbinder, Moran Feldman, Joseph Naor, and Roy Schwartz.
\newblock Submodular {Maximization} with {Cardinality} {Constraints}.
\newblock In Chandra Chekuri, editor, {\em Proceedings of the Twenty-Fifth
  Annual {ACM-SIAM} Symposium on Discrete Algorithms, {SODA} 2014, Portland,
  Oregon, USA, January 5-7, 2014}, pages 1433--1452. {SIAM}, 2014.
\newblock \href {https://doi.org/10.1137/1.9781611973402.106}
  {\path{doi:10.1137/1.9781611973402.106}}.

\bibitem{calinescu2011maximizing}
Gruia C{\u{a}}linescu, Chandra Chekuri, Martin P{\'{a}}l, and Jan
  Vondr{\'{a}}k.
\newblock Maximizing a monotone submodular function subject to a matroid
  constraint.
\newblock {\em {SIAM} Journal on Computing}, 40(6):1740--1766, 2011.
\newblock \href {https://doi.org/10.1137/080733991}
  {\path{doi:10.1137/080733991}}.

\bibitem{chakrabarti2015submodular}
Amit Chakrabarti and Sagar Kale.
\newblock Submodular maximization meets streaming: matchings, matroids, and
  more.
\newblock {\em Mathematical Programming}, 154(1-2):225--247, 2015.
\newblock \href {https://doi.org/10.1007/s10107-015-0900-7}
  {\path{doi:10.1007/s10107-015-0900-7}}.

\bibitem{chan2017online}
T-H.~Hubert Chan, Zhiyi Huang, Shaofeng H.-C. Jiang, Ning Kang, and
  Zhihao~Gavin Tang.
\newblock Online submodular maximization with free disposal: Randomization
  beats 1/4 for partition matroids.
\newblock In {\em Proceedings of the 28th Annual ACM-SIAM Symposium on Discrete
  Algorithms (SODA)}, SODA '17, pages 1204--1223, Philadelphia, PA, USA, 2017.
  Society for Industrial and Applied Mathematics.

\bibitem{chekuri2015streaming}
Chandra Chekuri, Shalmoli Gupta, and Kent Quanrud.
\newblock Streaming algorithms for submodular function maximization.
\newblock In Magn{\'u}s~M. Halld{\'o}rsson, Kazuo Iwama, Naoki Kobayashi, and
  Bettina Speckmann, editors, {\em Automata, Languages, and Programming}, pages
  318--330, Berlin, Heidelberg, 2015. Springer Berlin Heidelberg.

\bibitem{chekuri2015multiplicative}
Chandra Chekuri, T.~S. Jayram, and Jan Vondr{\'{a}}k.
\newblock On multiplicative weight updates for concave and submodular function
  maximization.
\newblock In Tim Roughgarden, editor, {\em Proceedings of the 2015 Conference
  on Innovations in Theoretical Computer Science (ITCS)}, pages 201--210.
  {ACM}, 2015.
\newblock \href {https://doi.org/10.1145/2688073.2688086}
  {\path{doi:10.1145/2688073.2688086}}.

\bibitem{ene2016constrained}
Alina Ene and Huy~L. Nguyen.
\newblock Constrained submodular maximization: Beyond 1/e.
\newblock In {\em {IEEE} 57th Annual Symposium on Foundations of Computer
  Science (FOCS)}, pages 248--257, 2016.
\newblock \href {https://doi.org/10.1109/FOCS.2016.34}
  {\path{doi:10.1109/FOCS.2016.34}}.

\bibitem{feige1998threshold}
Uriel Feige.
\newblock A threshold of $\ln n$ for approximating set cover.
\newblock {\em Journal of the ACM (JACM)}, 45(4):634--652, 1998.

\bibitem{feige2011maximizing}
Uriel Feige, Vahab~S. Mirrokni, and Jan Vondr{\'{a}}k.
\newblock Maximizing non-monotone submodular functions.
\newblock {\em {SIAM} Journal on Computing}, 40(4):1133--1153, 2011.
\newblock \href {https://doi.org/10.1137/090779346}
  {\path{doi:10.1137/090779346}}.

\bibitem{feldman_2020_one-way}
M.~Feldman, A.~Norouzi-Fard, O.~Svensson, and R.~Zenklusen.
\newblock The one-way communication complexity of submodular maximization with
  applications to streaming and robustness.
\newblock In {\em Proceedings of the 52nd Annual ACM SIGACT Symposium on Theory
  on Computing (STOC)}, pages 1363--1374, 2020.

\bibitem{feldman2018do}
Moran Feldman, Amin Karbasi, and Ehsan Kazemi.
\newblock Do less, get more: Streaming submodular maximization with
  subsampling.
\newblock In Samy Bengio, Hanna~M. Wallach, Hugo Larochelle, Kristen Grauman,
  Nicol{\`{o}} Cesa{-}Bianchi, and Roman Garnett, editors, {\em Advances in
  Neural Information Processing Systems (NeurIPS)}, pages 730--740, 2018.
\newblock URL:
  \url{https://proceedings.neurips.cc/paper/2018/hash/d1f255a373a3cef72e03aa9d980c7eca-Abstract.html}.

\bibitem{feldman2011unified}
Moran Feldman, Joseph Naor, and Roy Schwartz.
\newblock A unified continuous greedy algorithm for submodular maximization.
\newblock In {\em {IEEE} 52nd Annual Symposium on Foundations of Computer
  Science (FOCS)}, pages 570--579, 2011.
\newblock \href {https://doi.org/10.1109/FOCS.2011.46}
  {\path{doi:10.1109/FOCS.2011.46}}.

\bibitem{FilmusW14}
Yuval Filmus and Justin Ward.
\newblock Monotone submodular maximization over a matroid via non-oblivious
  local search.
\newblock {\em {SIAM} Journal on Computing}, 43(2):514--542, 2014.

\bibitem{gharan2011submodular}
Shayan~Oveis Gharan and Jan Vondr{\'{a}}k.
\newblock Submodular maximization by simulated annealing.
\newblock In {\em Proceedings of the Twenty-Second Annual {ACM-SIAM} Symposium
  on Discrete Algorithms (SODA)}, pages 1098--1116, 2011.
\newblock \href {https://doi.org/10.1137/1.9781611973082.83}
  {\path{doi:10.1137/1.9781611973082.83}}.

\bibitem{haba2020streaming}
Ran Haba, Ehsan Kazemi, Moran Feldman, and Amin Karbasi.
\newblock Streaming submodular maximization under a k-set system constraint.
\newblock In {\em Proceedings of the 37th International Conference on Machine
  Learning (ICML)}, pages 3939--3949, 2020.
\newblock URL: \url{http://proceedings.mlr.press/v119/haba20a.html}.

\bibitem{abs-1802-06212}
Chien{-}Chung Huang and Naonori Kakimura.
\newblock Multi-pass streaming algorithms for monotone submodular function
  maximization.
\newblock {\em CoRR}, abs/1802.06212, 2018.

\bibitem{HuangTW20}
Chien{-}Chung Huang, Theophile Thiery, and Justin Ward.
\newblock Improved multi-pass streaming algorithms for submodular maximization
  with matroid constraints, 2021.
\newblock URL: \url{https://arxiv.org/abs/2102.09679}, \href
  {http://arxiv.org/abs/2102.09679} {\path{arXiv:2102.09679}}, \href
  {https://doi.org/10.4230/LIPIcs.APPROX/RANDOM.2020.62}
  {\path{doi:10.4230/LIPIcs.APPROX/RANDOM.2020.62}}.

\bibitem{SubmodularWWW}
Andreas Krause.
\newblock Submodularity in machine learning.
\newblock \url{http://submodularity.org/}.

\bibitem{DBLP:conf/stoc/LeeMNS09}
Jon Lee, Vahab~S. Mirrokni, Viswanath Nagarajan, and Maxim Sviridenko.
\newblock Non-monotone submodular maximization under matroid and knapsack
  constraints.
\newblock In Michael Mitzenmacher, editor, {\em Proceedings of the 41st Annual
  {ACM} Symposium on Theory of Computing, {STOC} 2009, Bethesda, MD, USA, May
  31 - June 2, 2009}, pages 323--332. {ACM}, 2009.
\newblock \href {https://doi.org/10.1145/1536414.1536459}
  {\path{doi:10.1145/1536414.1536459}}.

\bibitem{LRVZ20}
Paul Liu, Aviad Rubinstein, Jan Vondr{\'{a}}k, and Junyao Zhao.
\newblock Cardinality constrained submodular maximization for random streams.
\newblock In {\em Advances in Neural Information Processing Systems (NeurIPS)},
  2021.
\newblock URL: \url{https://arxiv.org/abs/2111.07217}.

\bibitem{mcgregor2019better}
Andrew McGregor and Hoa~T. Vu.
\newblock Better streaming algorithms for the maximum coverage problem.
\newblock {\em Theory of Computing Systems}, 63(7):1595--1619, 2019.
\newblock \href {https://doi.org/10.1007/s00224-018-9878-x}
  {\path{doi:10.1007/s00224-018-9878-x}}.

\bibitem{MirzasoleimanBK16}
Baharan Mirzasoleiman, Ashwinkumar Badanidiyuru, and Amin Karbasi.
\newblock Fast constrained submodular maximization: Personalized data
  summarization.
\newblock In {\em Proceedings of the 33nd International Conference on Machine
  Learning, {ICML}}, pages 1358--1367, 2016.

\bibitem{mirzasoleiman2018streaming}
Baharan Mirzasoleiman, Stefanie Jegelka, and Andreas Krause.
\newblock Streaming non-monotone submodular maximization: Personalized video
  summarization on the fly.
\newblock In {\em Proceedings of the Thirty-Second {AAAI} Conference on
  Artificial Intelligence (AAAI-18)}, pages 1379--1386, 2018.
\newblock URL:
  \url{https://www.aaai.org/ocs/index.php/AAAI/AAAI18/paper/view/17014}.

\bibitem{nemhauser1978best}
George~L. Nemhauser and Laurence~A. Wolsey.
\newblock Best algorithms for approximating the maximum of a submodular set
  function.
\newblock {\em Mathematics of Operations Research}, 3(3):177--188, 1978.

\bibitem{nemhauser1978analysis}
George~L. Nemhauser, Laurence~A. Wolsey, and Marshall~L. Fisher.
\newblock An analysis of approximations for maximizing submodular set
  functions--—{I}.
\newblock {\em Mathematical Programming}, 14(1):265--294, 1978.

\bibitem{Norouzi-FardTMZ18}
Ashkan Norouzi{-}Fard, Jakub Tarnawski, Slobodan Mitrovic, Amir Zandieh,
  Aidasadat Mousavifar, and Ola Svensson.
\newblock Beyond 1/2-approximation for submodular maximization on massive data
  streams.
\newblock In Jennifer~G. Dy and Andreas Krause, editors, {\em Proceedings of
  the 35th International Conference on Machine Learning (ICML)}, volume~80 of
  {\em Proceedings of Machine Learning Research}, pages 3826--3835. {PMLR},
  2018.
\newblock URL: \url{http://proceedings.mlr.press/v80/norouzi-fard18a.html}.

\bibitem{schrijver2003combinatorial}
Alexander Schrijver.
\newblock {\em Combinatorial Optimization : Polyhedra and Efficiency}.
\newblock Springer-Verlag Berlin Heidelberg, 2003.

\bibitem{S20b}
Mohammad Shadravan.
\newblock Improved submodular secretary problem with shortlists, 2020.
\newblock URL: \url{https://arxiv.org/abs/2010.01901}, \href
  {http://arxiv.org/abs/2010.01901} {\path{arXiv:2010.01901}}.

\bibitem{S20}
Mohammad Shadravan.
\newblock Submodular {Matroid} {Secretary} {Problem} with {Shortlists}, January
  2020.
\newblock URL: \url{http://arxiv.org/abs/2001.00894}.

\bibitem{badanidiyuru2011buyback}
Ashwinkumar~Badanidiyuru Varadaraja.
\newblock Buyback problem - approximate matroid intersection with cancellation
  costs.
\newblock In {\em International Colloquium on Automata, Languages and
  Programming (ICALP)}, pages 379--390, 2011.
\newblock \href {https://doi.org/10.1007/978-3-642-22006-7\_32}
  {\path{doi:10.1007/978-3-642-22006-7\_32}}.

\bibitem{vondrak_2013_symmetry}
J.~Vondr\'{a}k.
\newblock Symmetry and approximability of submodular maximization problems.
\newblock {\em SIAM Journal on Computing}, 42(1):265--304, 2013.

\end{thebibliography}

\newpage
\appendix

\section{Single-Pass for Non-Monotone Submodular Functions}
\label{app:singleMatNonMon}

We now show how our single-pass algorithm for \MSMM, i.e., \cref{alg:singleMatMonotone}, can be extended to the non-monotone case, i.e., to \SMM, to obtain \cref{thm:mainSinglaMatNonMon}.
Only minor modifications are needed to the algorithm.
The modified algorithm appears as \cref{alg:singleMatNonMonotone}.
There are two changes compared to \cref{alg:singleMatMonotone}, our algorithm for the monotone case.
First, the parameter $\alpha$ is chosen differently, and is set to be (approximately) the solution to $e^\alpha = \frac{\alpha^2 + 2\alpha - 1}{\alpha - 1}$; more precisely, we set $\alpha= 1.9532$.
Second, when updating coordinates of the vectors $a_i$, we only increase a coordinate, corresponding to some element $u\in \cN$, if the total increase of the coordinate $u$ so far does not exceed some target value $p\in (0,1)$, which is set to $\frac{1}{m(c-1)+1}$.

\begingroup
\begin{algorithm}[ht]
\caption{Single-Pass Semi-Streaming Algorithm for {\MSMM}} \label{alg:singleMatNonMonotone}
\begin{algorithmic}[1]
\State Set $\alpha = 1.9532$, $m=\left\lceil\frac{3 \alpha}{\eps}\right\rceil$, $c=\frac{m}{m-\alpha}$, $p=\frac{1}{m(c-1)+1}$, and $L = \left\lceil \log_c(\frac{2 c}{\eps(c-1)}) \right\rceil$.
\State Set $a = 0 \in [0, 1]^\cN$ to be the zero vector, and let $b=-\infty$.

\For{every element arriving $u \in \cN$, if $\partial_u F(a) > 0$}\label{algline:loopOverUNonMon} 
	\State Let $i(u) = \left\lfloor \log_c(\partial_u F(a)) \right\rfloor$.\label{algline:index_i_u_nonMon}
	\Comment{%
	\raggedright Thus, $i(u)$ is largest index $i\in \mathbb{Z}$ with $c^i \leq \partial_u F(a)$.%
	}
  \State $\beta \coloneqq \max\{b, i(u) - \rank(M) - L\}$.
	\For{$i = \beta$ \textbf{to} $i(u)$}\label{algline:loopOverINonMon}
		\If{$A_i + u \in \cI$ \textbf{and} $\sum_{i=\beta}^{i(u)} a_i(u) \leq p$}\label{algline:add_u_test}
			\State $A_i \gets A_i + u$.
			\State $a_i \gets a_i + \frac{c^i}{m\cdot \partial_u F(a)} \characteristic_u$.\label{algline:modify_ai_nonMon}
		\EndIf
	\EndFor

\State Set $b \gets h - L$, where $h$ is largest index $i\in \mathbb{Z}$ satisfying $\sum_{j=i}^\infty |A_j| \geq \rank(M)$.

\State $a \gets \sum_{i=b}^{\infty} a_i$.

\State Delete from memory all sets $A_i$ and vectors $a_i$ with $i\in \mathbb{Z}_{<b}$.\label{algline:reduceMem_nonMon}

\EndFor

\State Set $S_k \gets \varnothing$ for $k \in \{0,\ldots, m-1\}$.

\State Let $q$ be largest index $i\in \mathbb{Z}$ with $A_i\neq \varnothing$.\label{algline:define_q_nonMon}

\For{$i = q$ \textbf{to} $b$ (stepping down by $1$ at each iteration)}\label{algline:build_Sk_nonMon}
	\While{$\exists u \in A_i \setminus S_{(i \bmod m)}$ with $S_{(i \bmod m)} + u\in \cI$}
		\State $S_{(i \bmod m)} \gets S_{(i \bmod m)} + u$.
	\EndWhile
\EndFor
\State \Return{a rounding $R\in \mathcal{I}$ of the fractional solution $s\coloneqq \frac{1}{m}\sum_{k=0}^{m-1} \characteristic_{S_k}$ with $f(R) \geq F(s)$}.\label{algline:returnNonMonotSubmodMat}
\end{algorithmic}
\end{algorithm}
\endgroup

We recall that most results shown in \cref{sec:singlepass} did not rely on monotonicity of $f$.
More precisely, the first result needing monotonicity of $f$ was \cref{lem:diffOptToAall}.
Moreover, also \cref{obs:boundDerivativesFAall} does not hold anymore, because it used the fact that we add $u$ to sets $A_i$ as long as the marginal contribution is large enough and $A_i + u$ is independent.
However, this is not always the case in \cref{alg:singleMatNonMonotone} because we will stop adding $u$ to sets $A_i$ if the condition $\sum_{i=\beta}^{i(u)} a_i(u) \leq p$ on Line~\ref{algline:add_u_test} of \cref{alg:singleMatNonMonotone} fails.

To circumvent this issue, we provide two bounds on the partial derivative $\partial_u F(\aall)$, one for elements $u\in \cN$ with $\aall(u)\leq p$ and one for elements $u\in \cN$ with $\aall(u) \geq p$.
In the first case, the condition $\sum_{i=\beta}^{i(u)}\leq p$ was always fulfilled, and we can replicate the same analysis as in the monotone case.
In the second case, we exploit that $\aall(u)\geq p$ is large to provide a bound on $\partial_u F(\aall)$ that depends on $p$.

Recall that $\ell(u)$ was defined, in \cref{sec:singlepass}, as largest index $i\in \mathbb{Z}$ such that $u\in \Span(\overline{A}_{i})$ (or $-\infty$, if no such index exists).
\begin{observation}\label{obs:boundDerivatiesSmallP}
\begin{equation*}
\partial_u F(\aall) \leq c^{\ell(u)+1} \qquad \forall u\in \cN \text{ with } \aall(u) \leq p\enspace.
\end{equation*}
\end{observation}
\begin{proof}
Since $u\not\in \Span(\overline{A}_{\ell(u)+1})$ and $\aall(u) \leq p$, we get that $u$ did not get added to the set $A_{\ell(u)+1}$ in \cref{alg:singleMatMonotone}, even though it fulfilled both $A_{\ell(u)+1} + u \in \cI$---because $A_{\ell(u)+1} + u \subseteq \overline{A}_{\ell(u)+1} + u \in \cI$---and $\sum_{i=b}^{i(u)} a_i(u)\leq \aall(u) \leq p$ when it got considered.
Hence, when $u$ got considered in Line~\ref{algline:loopOverUNonMon} of \cref{alg:singleMatNonMonotone}, we had $\partial_u F(a) < c^{\ell(u)+1}$.
Finally, by submodularity of $f$ and because $a \leq \aall$ (coordinate-wise), we have $\partial_u F(\aall) \leq \partial_u F(a) \leq c^{\ell(u)+1}$.
\end{proof}

\begin{lemma}\label{lem:boundDerivatiesLargeP}
\begin{equation*}
\partial_u F(\aall) \leq \frac{1}{1-p} \cdot c^{\ell(u)+1} \qquad \forall u\in \cN \text{ with } \aall(u) \geq p \enspace.
\end{equation*}
\end{lemma}
\begin{proof}
Let $u\in \cN$ with $\aall(u) \geq p$, and let $a$ be the vector at the beginning of the iteration of the for loop in Line~\ref{algline:loopOverUNonMon} of \cref{alg:singleMatNonMonotone} when $u$ got considered.
Moreover, let $\beta\coloneqq\max\{b, i(u) - \rank(M) -L\}$ be the $\beta$ computed and used at that iteration.
Because the multilinear extension $F$ is linear in each single coordinate, and in particular the one corresponding to $u$,
\begin{equation}\label{eq:contribution_of_u_in_over_p_case}
F\left(a + \sum_{i=\beta}^{\ell(u)} \frac{c^i}{m\cdot \partial_u F(a)} \characteristic_u\right) - F(a) = \frac{1}{m} \sum_{i=\beta}^{\ell(u)} c^i\enspace.
\end{equation}
Note that because the values of $a_i(u)$ are only increased during the iteration of the for loop in Line~\ref{algline:loopOverUNonMon} that considers $u$, we have
\begin{equation}\label{eq:aall_u_as_sum_of_cis}
\aall(u) = \sum_{i=\beta}^{\ell(u)} \frac{c^i}{m\cdot \partial_u F(a)}\enspace.
\end{equation}
Due to the same reason, we have $a(u)=0$. Hence, $\aall \wedge \characteristic_{\cN - u} \geq a$ (coordinate-wise).
We thus obtain
\begin{equation}\label{eq:boundDerLargePSumCi}
\begin{aligned}
\frac{1}{m} \sum_{i=\beta}^{\ell(u)} c^i
  &= F\left(a + \aall(u)\cdot \characteristic_u\right) - F(a) \\
  &\geq F\left((\aall\wedge \characteristic_{\cN-u}) + \aall(u)\cdot \characteristic_u\right) - F(\aall \wedge \characteristic_{\cN-u}) \\
  &= \aall(u) \cdot \partial_u F(\aall) \\
  &\geq p \cdot \partial_u F(\aall)\enspace,
\end{aligned}
\end{equation}
where the first equality is due to~\cref{eq:contribution_of_u_in_over_p_case,eq:aall_u_as_sum_of_cis}, the first inequality follows from submodularity of $f$ and $\aall\wedge \characteristic_{\cN-u} \geq a$, the second equality is due to multilinearity of $F$, and the last inequality holds because we are considering an element $u\in \cN$ with $\aall(u) \geq p$.
The result now follows due to
\begin{equation*}
\partial_u F(\aall)
  \leq \frac{1}{m p} \sum_{i=\beta}^{\ell(u)} c^i 
  \leq \frac{1}{m p} \sum_{i=-\infty}^{\ell(u)} c^i  
  = \frac{1}{mp(c-1)} \cdot c^{\ell(u) + 1}
  = \frac{1}{1-p} \cdot c^{\ell(u) + 1}\enspace,
\end{equation*}
where the first inequality is due to~\cref{eq:boundDerLargePSumCi}, and the last equality follows from our definition of $p$, i.e., $p\coloneqq \frac{1}{m(c-1)+1}$.
\end{proof}

We now combine \cref{obs:boundDerivatiesSmallP,lem:boundDerivatiesLargeP} to obtain a result analogous to \cref{lem:diffOptToAall} (which we had in the monotone case).
\begin{lemma}\label{lem:optDerivativesNonMon}
\begin{equation*}
F(\aall \vee \characteristic_{\OPT}) - F(\aall) \leq \sum_{u\in \OPT} c^{\ell(u) + 1} \enspace.
\end{equation*}
\end{lemma}
\begin{proof}
Let
\begin{equation*}
\OPTbig \coloneqq \left\{u\in \OPT \colon \aall(u) \geq p\right\}\enspace.
\end{equation*}
The result now follows from
\begin{align*}
F(\aall\vee {}&\characteristic_{\OPT}) - F(\aall)
  \leq \sum_{u\in \OPT} \partial_u F(\aall) \cdot (1-\aall(u)) \\
  &\leq \frac{1}{1-p} \sum_{u\in \OPTbig} c^{\ell(u)+1}\cdot (1-\aall(u)) + \sum_{u\in \OPT\setminus \OPTbig} c^{\ell(u)+1}\cdot (1-\aall(u))\\
  &\leq \sum_{u\in \OPT} c^{\ell(u) + 1}\enspace,
\end{align*}
where the first inequality uses concavity of $F$ along non-negative directions, the second one is due to \cref{obs:boundDerivatiesSmallP,lem:boundDerivatiesLargeP}, and the last one holds since $1-\aall(u)\leq 1-p$ in the first sum (because $u\in \OPTbig$) and $1-\aall(u) \leq 1$ in the second sum.
\end{proof}

As in the monotone case, we now would like to relate the difference between $f(\OPT)$ and $F(\aall)$ to the above-derived bounds on $\partial_u F(\aall)$.
\cref{lem:diffOptToAall}, which we used in the monotone case, does not hold for non-monotone functions $f$.
To avoid the need for monotonicity, we bound the difference $F(\aall \vee \characteristic_{\OPT})-F(\aall)$ instead.
To relate $F(\aall \vee \characteristic_{\OPT})$ to $f(\OPT)$, we exploit that $\aall$ has small coordinates, through the following known lemma (for completeness, we prove this lemma here, however, we note that it can also be viewed as an immediate corollary of either Lemma~7 of~\cite{chekuri2015multiplicative} or Lemma~2.2 of~\cite{buchbinder2014submodular}).
\begin{lemma}\label{lem:multlin_value_small_coords}
Let $f: 2^\cN\to \mathbb{R}_{\geq 0}$ be a non-negative submodular function with multilinear extension $F$, and let $p\in [0,1]$, $x\in [0,p]^\cN$, and $S\subseteq \cN$. Then
$F(x \vee \characteristic_S) \geq (1-p)f(S)$.
\end{lemma}
\begin{proof}
We use the fact that the multilinear extension is lower bounded by the Lov\'asz extension $f_L\colon [0,1]^\cN \to \mathbb{R}_{\geq 0}$, which is given by
\begin{equation*}
f_L(y)\coloneqq \int_{t=0}^1 f\left(\{u\in \cN \colon y(u) > t\}\right) dt \qquad \forall y\in [0,1]^\cN\enspace.
\end{equation*}
Hence, $F(y) \geq f_L(y)$ for all $y\in [0,1]$ (see, e.g.,~\cite{vondrak_2013_symmetry} for a formal proof of this well-known fact). The result now follows from
\begin{align*}
F(x\vee \characteristic_S) &\geq f_L(x\vee \characteristic_S) \\
                           &= \int_{t=0}^1 f(S \cup \{u\in \cN\colon x(u) > t\}) dt \\
                           &\geq \int_{t=p}^1 f(S \cup \{u\in \cN \colon x(u) > t\}) dt \\
							             &= (1-p) f(S)\enspace,
\end{align*}
where the last equality uses that $x\in [0,p]^\cN$.
\end{proof}

By applying \cref{lem:multlin_value_small_coords} in our context, we get the following lower bound on $F(\aall \vee \characteristic_{\OPT})$ in terms of $f(\OPT)$.
\begin{corollary}\label{cor:boundOptNonMon}
\begin{equation*}
\left(1-p-\frac{1}{m}\right)\cdot f(\OPT) \leq F(\aall \vee \characteristic_{\OPT})
\end{equation*}
\end{corollary}
\begin{proof}
This is an immediate consequence of \cref{lem:multlin_value_small_coords} and the fact that $\aall\in [0,p+\frac{1}{m}]$, which holds due to the following.
For any element $u\in \cN$, \cref{alg:singleMatNonMonotone} does not continue to increase coordinates $a_i(u)$ if the sum of the $a_j(u)$ already surpasses $p$.
Moreover, every increase of $u$ happens through an update of one of the $a_i$ vectors by increasing $a_i(u)$ by $\frac{c^i}{m\cdot \partial_u F(a)}\leq \frac{1}{m}$, because $c^i \leq c^{i(u)} \leq \partial_u F(a)$ by choice of $i(u)$.
\end{proof}

Finally, by combining the above obtained relations, and using our choices of the parameters $\alpha$, $c$, $p$, and $m$, we obtain the desired result.
Note that the bound on the memory requirement of \cref{alg:singleMatMonotone} also holds for \cref{alg:singleMatNonMonotone}, as it is unrelated to monotonicity of $f$ or to the minor differences between the two algorithms.
\begin{proof}[Proof of \cref{thm:mainSinglaMatNonMon}]
Because $f(R) \geq F(s)$, it suffices to show that $F(s) \geq 0.1921 \cdot f(\OPT)$.
The value of $\OPT$ and $F(s)$ can be related as follows:
\begin{equation*}
\begin{aligned}
\left(1-p-\frac{1}{m}\right)\cdot f(\OPT)
  &\leq F(\aall \vee \characteristic_{\OPT}) \\
  &\leq F(\aall) + \sum_{u\in \OPT} c^{\ell(u)+1} \\
  &\leq \frac{1}{1-c^{-m} - \frac{\eps}{2c}} \cdot F(s) + (c-1)\cdot \left(1+\frac{\eps}{2c}\right)\cdot \frac{m}{1-c^{-m}} \cdot F(s) \\
  &\leq \left(m\cdot (c-1)\cdot \left(1+\frac{\eps}{2c}\right) + 1\right) \cdot \frac{1}{1-c^{-m}-\frac{\eps}{2c}} \cdot F(s)\enspace,
\end{aligned}
\end{equation*}
where the first inequality is due to \cref{cor:boundOptNonMon},
the second one follows from \cref{lem:optDerivativesNonMon},
and the third one is implied by \cref{lem:lowerBoundFsInFAall,lem:upperBoundSumCOpt} (we recall that these results did not need monotonicity of $f$).

Regrouping terms in the above inequality and simplifying, we obtain the following:\footnote{The first steps of the derivation are analogous to the ones performed in the proof of \cref{thm:mainSingleMatMon}. The only difference is the additional term of $(1-p-\sfrac{1}{m})$.}
{\allowdisplaybreaks
\begin{align}\label{eq:lowerBoundFsInOptUsingAlphaNonMon}
F(s) &\geq \frac{1-c^{-m}-\frac{\eps}{2c}}{m\cdot (c-1) \cdot \left(1+\frac{\eps}{2c}\right) + 1} \cdot \left(1-p-\frac{1}{m}\right)\cdot f(\OPT) \\\nonumber
     &\geq \frac{1-e^{-\alpha} - \frac{\eps}{2c}}{\alpha\cdot \left(c + \frac{\eps}{2}\right) + 1} \cdot \left(1-p-\frac{1}{m}\right)\cdot f(\OPT) \\\nonumber
     &\geq \left( \frac{1-e^{-\alpha}}{(\alpha + 1)\cdot \left(c+\frac{\eps}{2}\right)} - \frac{\eps}{2c} \right)\cdot \left(1-p-\frac{1}{m}\right)\cdot f(\OPT) \\\nonumber
     &\geq \left(\frac{1-e^{-\alpha}}{\alpha+1} \cdot \frac{1}{1+\eps} - \frac{\eps}{2c}\right) \cdot \left(1-p-\frac{1}{m}\right)\cdot f(\OPT) \\\nonumber
     &\geq \left(\frac{1-e^{-\alpha}}{\alpha+1} - \eps \right) \cdot \left(1-p-\frac{1}{m}\right)\cdot f(\OPT) \\\nonumber
     &\geq \left(\frac{1-e^{-\alpha}}{\alpha+1}\cdot (1-p) - \eps \right) \cdot f(\OPT) \\\nonumber
     &= \left(\frac{1-e^{-\alpha}}{\alpha+1}\cdot \frac{\alpha c}{\alpha c + 1} - \eps \right) \cdot f(\OPT) \\\nonumber
     &\geq \left(\frac{(1-e^{-\alpha})\cdot \alpha}{(\alpha+1)^2} - \eps \right) \cdot f(\OPT)\enspace,
\end{align}}
where the different inequalities hold due to the following.
The second inequality uses that $c=\frac{m}{m-\alpha}$, which implies $c^{-m} = (1-\sfrac{\alpha}{m})^m \leq e^{-\alpha}$ and $m (c-1) = \alpha c$.
The third inequality holds because $c+\sfrac{\eps}{2} \geq 1$ and $\alpha\cdot (c+\sfrac{\eps}{2}) + 1 \geq 1$.
The forth one follows from $c = \frac{m}{m-\alpha} = (1-\sfrac{\alpha}{m})^{-1} \leq (1-\sfrac{\eps}{3})^{-1} \leq 1+\sfrac{\eps}{2}$ by using our definitions of $c$ and $m$.
The fifth inequality uses that $(1+\eps)^{-1}\geq 1-\eps$ and $\frac{1-e^{-\alpha}}{\alpha+1}\leq \frac{1}{2}$.
The sixth inequality holds because $\frac{1-e^{-\alpha}}{\alpha+1}\frac{1}{m}\leq \frac{1}{3m} \leq \frac{\eps}{9\alpha}$ and $p\cdot \eps = \frac{\eps}{\alpha c + 1} \geq \frac{\eps}{9\alpha}$.
The requality uses that $p=\frac{1}{\alpha c + 1}$.
Finally, the last inequality follows from $c \geq 1$.

The claimed approximation factor of $1.921$ is obtained by plugging in our value of $\alpha=1.9532$ (for a small enough $\eps > 0$).
\end{proof}

\section{Proof of \texorpdfstring{\cref{lem:at-least-prob}}{Lemma~\ref*{lem:at-least-prob}}} \label{sec:random_appendix}

In this section we prove \cref{lem:at-least-prob}. We begin with the following helper lemma.
\begin{lemma}
\label{lem:uniform-partition}
Suppose the element of $\cN$ appear in the stream in a uniformly random order, and we partition $\cN$ by \Cref{alg:partition} into $\alpha k$ windows, then this is equivalent to assigning each $u \in \cN$ to one of $\alpha k$ different buckets uniformly and independently at random.
\end{lemma}

\begin{proof}
The way we define the window sizes $n_1, n_2, \dotsc,$ is equivalent to placing each element independently into a random bucket, and then letting $n_i$ be the number of elements that ended up in bucket $i$. Hence, the distribution of the window sizes is correct. Furthermore, conditioned on the window sizes, each window is simply assigned the elements in some positions of the random stream, and therefore, the set of elements it gets is a uniformly random subset of $\cN$ of the right size which is independent of the partitioning of the remaining elements between the other windows.
\end{proof}

Let us denote now by $R$ the random coins used in \cref{line:R_selection} of \cref{alg:matroid-basic}. Below, we prove \cref{lem:at-least-prob} conditioned in a fixed choice of $R$, which implies that the lemma holds also unconditionally due to the law of total probability. 
Let $J_{u}(\cH_{i-1}, R)$ be the set of indices $j$ where there exists some partition $P$ that implies the history $\cH_{i-1}$ given $R$\footnote{Observe that once $R$ and $P$ are fixed, \cref{alg:matroid-basic} becomes deterministic, and therefore, $P$ and $R$ determine $\cH_{i - 1}$.} and has $P(u) = j$.
\begin{lemma} \label{lem:equal_probability_possible}
Conditioned on history $\cH_{i - 1}$ and random coins $R$, the probability of an element $u \in \cN \setminus \cH_{i - 1}$ to end up in every window corresponding to the indices of $J_u(\cH_{i - 1})$ is equal. Furthermore, $J_u(\cH_{i - 1}, R)$ includes every integer $i \leq j \leq \alpha k$
\end{lemma}
\begin{proof}
Choose any $j, j^\prime \in J_u(\cH_{i - 1})$. We would like to show that for each partition $P$ that implies the history $\cH_{i - 1}$ given $R$ and has $P(u) = j$ we can create another partition $\tilde{P}$ that implies $\cH_{i-1}$ given $R$ by setting $\tilde{P}(u) = j'$ and keeping all other values of $\tilde{P}$ as in $P$. Since $\tilde{P}$ is equal to $P$ everywhere except on $u$, this maps each such partition $P$ to a unique partition $\tilde{P}$, establishing that the number of partitions that imply $\cH_{i-1}$ given $R$ and map $u$ to $w_j$ is not larger than the number of such partitions mapping $u$ to $w_{j'}$. Since this is true for every $j, j \in J_u(\cH_{i - 1}, R)$, and all the partitions have equal probability by \cref{lem:uniform-partition}, the first part of the lemma follows once we show the above.

Since $u \notin \cH_{i - 1}$, for an index $j$ to be in $J_u(\cH_{i - 1}, R)$, one of two things must happen. The first option is that $j \geq i$, in which case trivially \cref{alg:matroid-basic} could not add $j$ to $H$ while processing the first $i - 1$ windows (note that the existence of this option already implies the second part of the lemma). The second option is that $j < i$, but $u$ was not selected by \cref{alg:matroid-basic} when it arrived because either $u$ was never the maximum element found in Line~\ref{algline:argmax}, or if it was, its marginal value was not sufficient to replace the current solution. In all these cases, removing or adding $u$ to window $w_j$ does not change the history $\cH_{i-1}$. Thus, given that $P$ implies the history $\cH_{i - 1}$ given $R$, changing $P(u)$ from one index of $J_u(\cH_{i - 1}, R)$ to another does not change this history. 
\end{proof}

We are now ready to prove \cref{lem:at-least-prob}, which we repeat here for convenience.

\lemAtLeastProb*

\begin{proof}
Since $u$ must appear in some window, and it can appear only in windows whose indices appear in $J_u(\cH_{i - 1}, R)$, \cref{lem:equal_probability_possible} implies that conditioned on $R$ we have
\[
    1 = \mspace{-4mu} \sum_{j' \in J_u(\cH_{i - 1},R )} \mspace{-27mu} \Pr[u\in w_{j'} \mid \cH_{i-1}]
      = |J_u(\cH_{i - 1}, R)| \cdot \Pr[u \in w_j \mid \cH_{i-1}]
      \leq
      \alpha k \cdot \Pr[u \in w_j \mid \cH_{i-1}]
      \enspace.
\]
As mentioned above, the conditioning on $R$ can be dropped by the law of totol probability, which implies the lemma.
\end{proof}

\end{document}